\documentclass[preprint, 12pt]{elsarticle}
\journal{elsevier}
\usepackage{xr}
\usepackage{array,bm,amsmath,amssymb,mathrsfs,epsfig,epstopdf}
\usepackage{amsthm}
\usepackage{setspace} 
\usepackage{CJK}
\usepackage{color}
\usepackage{lscape}
\usepackage{verbatim}
\usepackage{enumerate}
\usepackage{booktabs}
\usepackage{threeparttable}
\usepackage{multicol} 
\usepackage{longtable}
\usepackage{multirow}
\usepackage{bm}
\usepackage{bbm}
\usepackage{cases}
\usepackage{caption}
\usepackage{natbib}
\usepackage{algorithm}
\usepackage{algorithmic}
\usepackage{geometry}
\usepackage{framed}
\usepackage[pdftex,bookmarksnumbered,bookmarksopen,colorlinks, citecolor=blue,linkcolor=blue,urlcolor=blue]{hyperref}
\usepackage{graphicx}
\usepackage{amsmath}
\usepackage{amsfonts}
\usepackage{float}
\usepackage{ragged2e}
\usepackage[figuresright]{rotating}
\usepackage{mathtools}
\usepackage{pifont}
\usepackage{indentfirst}  
\usepackage{listings}
\definecolor{codegreen}{rgb}{0,0.6,0}
\definecolor{codegray}{rgb}{0.5,0.5,0.5}
\definecolor{codepurple}{rgb}{0.58,0,0.82}
\definecolor{backcolour}{rgb}{0.95,0.95,0.92}
\lstdefinestyle{mystyle}{
backgroundcolor=\color{backcolour},   
commentstyle=\color{codegreen},
keywordstyle=\color{magenta},
numberstyle=\tiny\color{codegray},
stringstyle=\color{codepurple},
basicstyle=\ttfamily\footnotesize,
breakatwhitespace=false,         
breaklines=true,                 
captionpos=b,                    
keepspaces=true,                 
numbers=left,                    
numbersep=5pt,                  
showspaces=false,                
showstringspaces=false,
showtabs=false,                  
tabsize=2
}

\newcommand\gA{{\mathcal{A}}}

\newcommand\gS{{\mathcal{S}}}

\newcommand\E{\mathbb{E}}

\newcommand\PP{\mathbb{P}}

\newcommand\Var{\mathrm{Var}}

\DeclareMathOperator*{\argmin}{arg\,min}

\DeclareMathOperator{\sign}{sign}

\newtheorem{theorem}{Theorem}[section]
\newtheorem{lemma}{Lemma}[section]

\newtheorem{condition}{Condition}
\begin{document}
\title{\bf Covariance Matrix Estimation for High-Dimensional Interval-Valued Data with Positive Definiteness}

\author[1,2]{Wan Tian}
\ead{wantian61@foxmail.com}

\author[3]{Wenhao Cui}
\ead{wenhaocui1992@buaa.edu.cn}

\author[4]{Rui Zhang}
\ead{zhrui1993@gmail.com}

\author[5,6]{Bingyi Jing}
\ead{bingyijing@cuhk.edu.cn}

\author[7]{Yang Liu}
\ead{mseliuyang@ustc.edu.cn}

\author[8,9,10]{Yijie Peng\corref{cor1}}
\ead{pengyijie@pku.edu.cn}

\cortext[cor1]{Corresponding author}

\address[1]{Advanced Institute of Information Technology, Peking University}
\address[2]{Wangxuan Institute of Computer Technology, Peking University, China, 100871}
\address[3]{School of Economics and Management, Beihang University, Beijing 100191, China}
\address[4]{Center for Applied Statistics and School of Statistics,Renmin University of China, Beijing, China, 100872}
\address[5]{Department of Statistics and Data Science, Southern University of Science and Technology, Shenzhen, China, 518055}

\address[6]{School of Artificial Intelligence, The Chinese University of Hong Kong, Shenzhen, China, 518172 }

\address[7]{School of Public Affairs, University of Science and Technology of China, Hefei, China, 230026}
\address[8]{PKU-Wuhan Institute for Artificial Intelligence}
\address[9]{Xiangjiang Laboratory, Changsha 410000, China}
\address[10]{Guanghua School of Management, Peking University, Beijing,  China, 100871}

\begin{abstract}
In the realm of high-dimensional data analysis, the estimation of covariance matrices is a fundamental task, and this holds true for interval-valued data as well. However, there is no unified definition for the covariance matrix of interval-valued data, let alone established estimation methods in high-dimensional settings. This paper presents a novel approach to estimating covariance matrices for high-dimensional interval-valued data while ensuring positive definiteness. We begin by assuming that the upper and lower bounds of interval-valued variables share the same dependency structure. Based on this assumption, we extend the classical soft-thresholding covariance matrix estimator to the interval-valued scenario, referred to as the Interval-valued Soft-Thresholding (IST) estimator. Subsequently, to ensure the positive definiteness of the estimator, we impose a positive definiteness constraint on the IST estimator. We derive an alternating direction method to solve the proposed problem and establish its convergence. Under some very mild conditions, we develop a non-asymptotic statistical theory for the proposed estimator. Simulation studies and applications to high-frequency financial data from the CSI 300 Index demonstrated the effectiveness of the proposed estimator.

\end{abstract}

\maketitle
\section{Introduction}
In recent years, the modeling and analysis of interval-valued data have garnered increasing attention in the fields of econometrics and statistics \citep{billard2003statistics, sun2018threshold, le2012symbolic, sun2022model, tian2024minimum, gonzalez2013constrained}. Generally, there are three primary reasons for the generation of interval-valued data \citep{di2021tensor}. Firstly, data aggregation: compressing a large number of observations into intervals reduces data volume and enhances modeling efficiency. Secondly, privacy protection: replacing precise values with interval values. Lastly, uncertain observations: real-world measurements may be subject to inaccuracies due to environmental factors, limitations in sensory devices, or inconsistencies across data sources. Compared to point-valued data, modeling interval-valued data offers two distinct advantages \citep{han2016vector}. Firstly, interval-valued data inherently contains richer information regarding both central level and variability, thereby enabling more efficient statistical inference. Secondly, while specific disturbances or noise may have catastrophic effects on point-valued data modeling, these issues can be effectively mitigated through interval-valued data modeling approaches \citep{chou2005forecasting}.

Currently, research on interval-valued data can be categorized based on different methods of handling such data. The first approach represents intervals using representative points, commonly including upper bound, lower bound, center, and range, and then extends methods developed for point-valued data to the interval-valued setting. Major research topics in this category include regression \citep{diamond1990least, gil2002least, billard2002symbolic,neto2008centre, neto2010constrained} , classification \citep{palumbo1999non, rasson2000symbolic,qi2020interval}, Bayesian modeling \citep{hoff2009first, Zhang2019,Xu2022bayes}, and model averaging \citep{sun2022model}. The second approach focuses on selecting appropriate distance metrics between pairs of intervals, constructing corresponding estimators, and further developing statistical inference theories. These studies are primarily concentrated on time series modeling \citep{han2012autoregressive, han2016vector, sun2024nonparametric, sun2018threshold}.  We know that high-dimensional covariance matrix estimation is a fundamental task in statistical learning. However, for interval-valued data, not only is there no unified definition of the covariance matrix, but the estimation methods in high-dimensional settings are also highly limited. Below, we discuss the existing research on this topic.

In low-dimensional settings, to the best of our knowledge, there are four methods for estimating the covariance matrix of interval-valued data. \citet{cazes1997extension} proposed using the vertices and centers of intervals to estimate the covariance matrix. Clearly, both methods only utilize partial information of the intervals. \citet{wang2012cipca} assumed an infinite number of uniformly distributed points within the hypercube and defined the inner product operator for interval-valued variables, along with the corresponding covariance. \citet{bertrand2000descriptive} defined the empirical distribution function and density function for interval-valued data, providing the foundation for empirical mean and variance. Building on this work, \citet{billard2008sample} further extended the definition to the covariance of bivariate interval-valued variables. In addition, \citet{tian2024minimum} extended the minimum covariance determinant estimator to the interval-valued setting, thereby robustifying the covariance matrix proposed by \citet{billard2008sample}. They also proved the high breakdown point property of the estimator. 

In high-dimensional settings, research on matrix estimation for interval-valued data is extremely limited. The first significant work might be \citet{tian2024minimum}, which builds on the covariance matrix proposed by \citet{billard2008sample} and then introduces both regularization-based and projection-based estimators, extending them to high-dimensional scenarios to ensure positive definiteness. However, these two estimators are only applicable to moderate-dimensional cases and do not incorporate sparsity. There is also some work on high-dimensional sparse precision matrix estimation. \citet{wu2024idgm} first used bivariate point values to represent intervals, and then constructed an estimator based on graphical lasso \citep{friedman2008sparse}. This leads to the correlation between pairs of interval-valued variables being measured not by scalars, but by a block matrix. This not only doubles the dimensionality of the estimated precision matrix, but also lacks interpretability. Thus, it is evident that an effective estimation method for the covariance matrix of high-dimensional interval-valued data is currently unavailable. 


In this paper, we propose a novel framework for estimating the precision matrix of high-dimensional interval-valued data. Specifically, we first assume that the upper and lower bounds of interval-valued variables share the same dependency structure. Then, we extend the classic method of soft-thresholding estimator for high-dimensional point-valued data covariance matrix estimation to the interval-valued setting, which we refer to as the Interval-valued Soft-Thresholding (IST) estimator. To ensure that the estimator is positive definite, we impose a positive definite constraint on the IST estimator. On the optimization front, we derive an efficient alternating direction method to solve the problem and prove the convergence of the algorithm. Under some mild conditions, we develop non-asymptotic statistical theory for the proposed estimator. In simulation studies, we validate the effectiveness of the estimator across various data-generating processes (DGPs) and covariance matrix structures, while also analyzing the challenges associated with estimating matrices under different structural assumptions.

Empirically, we apply the proposed methodology to the estimation of covariance matrix using high-frequency financial data. The increasing availability of high-frequency financial data has spurred substantial interest in developing new methods for estimation and inference. High-frequency data refers to time series collected at an extremely fine scale, such as tick-by-tick market data, where each individual transaction is recorded as a ``tick”. Because of the large number of transactions, high-frequency datasets typically contain a vast amount of information. This abundance of observations allows for high statistical precision, which has motivated the development of a variety of specialized methods for analyzing such data. 

Substantial effort has been devoted to the estimation of return volatility, which plays a central role in asset pricing, risk management, hedging strategies, portfolio allocation, and forecasting. Volatility quantifies the dispersion of returns and serves as a key indicator of the risk associated with a financial asset. Various methods have been proposed for the estimation of return volatility matrix, which include the two-scale realized volatility estimator \citep{zhang2005tale}, the multi-scale realized volatility estimator \citep{zhang2006efficient}, the realized kernel \citep{barndorff2008designing}, the pre-averaging approach \citep{jacod2009microstructure}, and the quasi-maximum likelihood estimator \citep{xiu2010quasi}. 

However, existing methods often suffer from several shortcomings that limit their empirical effectiveness. First, high-frequency financial data, particularly tick-by-tick data, are difficult to obtain and typically come at a cost, making accessibility a major challenge. This issue is further compounded by the institutional evolution of equity markets: the proliferation of competing electronic trading networks increases the likelihood of slight timing discrepancies in reporting as well as other forms of pricing errors. Moreover, most existing approaches rely exclusively on return data, which are constructed using only the opening and closing prices within short intervals (e.g., five minutes). As a result, potentially valuable information, such as the highest and lowest prices observed within each interval, remains unused. 

The second challenge is that the estimation of high-dimensional volatility matrices using high-frequency financial data remains relatively underexplored. Most of the existing methods discussed above are developed for low-dimensional settings, with a few notable exceptions such as \cite{2012Vast,ait2017using,kong2018}. To the best of our knowledge, however, no existing work has utilized interval-valued data for the estimation of high-dimensional volatility matrices. The method proposed in this study therefore complements the current literature by explicitly incorporating the rich information embedded in high-frequency financial data.

In addition, candlestick data, which record the highest, lowest, opening, and closing prices over intervals such as five minutes, are widely available and generally accessible to researchers at a lower cost, and in some cases free of charge. This type of data naturally aligns with the framework considered in this study. By leveraging interval-valued high-frequency financial data, we empirically demonstrate the effectiveness of the proposed method. Our analysis shows that incorporating the full range of observed prices allows the model to capture richer patterns, which leads to economically meaningful conclusions, highlighting how the additional information embedded in the data range can be systematically exploited to improve inference and prediction.

The structure of the paper is as follows. Section \ref{sec2} introduces the proposed IST estimator and presents its optimization algorithm. Section \ref{sec3} proves the convergence of the algorithm and develops the non-asymptotic statistical theory for the IST estimator. Sections \ref{sec4} and \ref{sec5} present simulation studies and real data applications to demonstrate the effectiveness of the proposed estimator. Section \ref{sec6} concludes the paper. All the theoretical proofs are provided in \ref{appendix}.

\section{Methodology}  \label{sec2}
In this section, we introduce the proposed IST estimator and derive the alternating direction method to solve the optimization problem.

Let \(Y = (Y_{ij})_{1\leq i\leq n, 1\leq j\leq p} = (Y_{(1)}, Y_{(2)}, \cdots, Y_{(p)}) = (Y_1,Y_2,\cdots, Y_n)^\top\) be the matrix of interval-valued data observations, where \(Y_{(j)} = (Y_{1j}, Y_{2j}, \cdots, Y_{nj})^\top\) represents the \(j\)-th interval-valued variable with \(n\) observations, and \(Y_i = (Y_{i1},Y_{i2},\cdots, Y_{ip})^\top\) denotes a single observation of \(p\) interval-valued variables. Each interval-valued observation unit \(Y_{ij} = [Y^l_{ij}, Y^u_{ij}]\) consists of the upper bound \(Y^u_{ij}\) and lower bound \(Y^l_{ij}\). Our objective is to estimate the covariance matrix \(\Sigma = (\sigma_{ij})_{1\leq i,j\leq p}\) of the \(p\) interval-valued variables based on the observation matrix \(Y\), in order to capture the dependencies among the variables.

Without loss of generality, we assume that the upper and lower bounds of the \( p \) interval-valued variables share the same dependency structure, and both can be described by \(\Sigma\). In this case, we can directly extend the soft-thresholding estimator to the interval-valued setting, which leads to the following optimization problem:
\begin{equation} \label{directsoft}
\argmin_\Sigma \frac{1}{2}\lVert \Sigma - S^l\rVert^2_F + \frac{1}{2}\lVert \Sigma - S^u\rVert^2_F + \lambda \lVert \Sigma\rVert,
\end{equation}
where \( \|\cdot\|_F \) denotes the Frobenius norm of the matrix, and \( \lVert\Sigma \rVert = \sum_{i\neq j} |\sigma_{ij}|\). \(S^l = (s^l_{ij})_{1\leq i,j\leq p}\) and \(S^u = (s^u_{ij})_{1\leq i,j\leq p}\) represent the empirical covariance matrices corresponding to the upper bound observation matrix \(Y^u = (Y^u_{ij})_{1\leq i\leq n, 1\leq j\leq p}\) and the lower bound observation matrix \(Y^l = (Y^l_{ij})_{1\leq i\leq n, 1\leq j\leq p}\), respectively. 

However, the solution to optimization problem (\ref{directsoft}) can only ensure positive definiteness with high probability \citep{xue2012positive}; the resulting estimator is not guaranteed to be positive definite in practice. To address this issue, we impose a positive definiteness constraint on problem (\ref{directsoft}), which leads to the following formalized optimization problem:
\begin{equation}\label{IST}
\argmin_{\Sigma \succeq \epsilon I} \frac{1}{2}\lVert \Sigma - S^l\rVert^2_F + \frac{1}{2}\lVert \Sigma - S^u\rVert^2_F + \lambda \lVert \Sigma\rVert,
\end{equation}
where \(I\) represents the identity matrix, and \(\epsilon\) is an arbitrarily small positive number. It is important to note that \(\epsilon\) is not an turning parameter like \( \lambda \); it is merely introduced to ensure that the smallest eigenvalue of the estimator is no smaller than \(\epsilon\). We then discuss how to solve optimization problem (\ref{IST}) using the alternating direction method. We first introduce the consensus matrix \(\Gamma\) and transform the optimization objective (\ref{IST}) into the following equivalent form:
\begin{equation}\label{consensusIST}
\argmin_{\Sigma, \Gamma} \left\{
\frac{1}{2}\lVert \Sigma - S^l\rVert^2_F + \frac{1}{2}\lVert \Sigma - S^u\rVert^2_F + \lambda \lVert \Sigma\rVert: \Sigma = \Gamma, \Gamma\succeq \epsilon I.
\right\}
\end{equation}

The augmented Lagrangian form of the optimization problem (\ref{consensusIST}) is given by
\[
\ell(\Sigma, \Gamma, \Lambda) = \frac{1}{2}\lVert \Sigma - S^l\rVert^2_F + \frac{1}{2}\lVert \Sigma - S^u\rVert^2_F + \lambda \lVert \Sigma\rVert - \langle \Lambda, \Sigma - \Gamma \rangle + \frac{1}{2\beta} \lVert \Sigma - \Gamma\rVert^2_F,
\]
where \(\Lambda\) is the Lagrange multiplier matrix, and \(\beta\) is the penalty parameter. Given the results \(\Sigma^{(q)}\), \(\Gamma^{(q)}\), and \(\Lambda^{(q)}\) from the \(q\)-th iteration, the \((q+1)\)-th iteration consists of the following three steps:
\begin{align}
\Sigma \text{ step: } & \Sigma^{(q+1)} = \argmin_{\Sigma} \ell(\Sigma, \Gamma^{(q)}, \Lambda^{(q)}), \label{Sigmastep} \\
\Gamma \text{ step: } & \Gamma^{(q+1)} = \argmin_{\Gamma\succeq \epsilon I} \ell(\Sigma^{(q+1)}, \Gamma, \Lambda^{(q)}), \label{Gammastep} \\
\Lambda \text{ step: } & \Lambda^{(q+1)} = \argmin_{\Lambda} \ell(\Sigma^{(q+1)}, \Gamma^{(q+1)}, \Lambda). \label{Lambdastep}
\end{align}

We now discuss these three steps individually and provide their closed-form solutions. For \(\Sigma \text{ step}\), we have
\[
\begin{aligned}
\Sigma^{(q+1)} &= \argmin_{\Sigma} \ell(\Sigma, \Gamma^{(q)}, \Lambda^{(q)})\\
& = \argmin_{\Sigma} \frac{1}{2}\lVert \Sigma - S^l\rVert^2_F + \frac{1}{2}\lVert \Sigma - S^u\rVert^2_F + \lambda \lVert \Sigma\rVert - \langle \Lambda^{(q)}, \Sigma - \Gamma^{(q)} \rangle + \frac{1}{2\beta} \lVert \Sigma - \Gamma^{(q)}\rVert^2_F\\
& =\argmin_{\Sigma} \frac{1}{2}\lVert \Sigma - S^l\rVert^2_F + \frac{1}{2}\lVert \Sigma - S^u\rVert^2_F + \lambda \lVert \Sigma\rVert - \langle \Lambda^{(q)}, \Sigma\rangle + \frac{1}{2\beta} \lVert \Sigma - \Gamma^{(q)}\rVert^2_F\\
& = \argmin_{\Sigma} \frac{1}{2} \bigg \| \Sigma - \frac{\beta(S^l + S^u + \Lambda^{(q)}) + \Gamma^{(q)}}{2\beta+1} \bigg \|^2_F+  \frac{\beta \lambda}{2\beta+1} \lVert \Sigma\rVert\\
& = \frac{1}{2\beta+1} \gS(\beta(S^l + S^u + \Lambda^{(q)}) + \Gamma^{(q)}, \beta \lambda),
\end{aligned}
\]
where \(\gS\) is an element-wise soft thresholding operator applied to a matrix. For example, for any matrix \(G = (g_{ij})_{1\leq i,j\leq p}\) and threshold \(\eta\), we have \(\gS(G, \eta) = (s(g_{ij}, \eta))_{1\leq i,j\leq p}\), whth
\[
s(g_{ij}, \eta) = \sign(g_{ij}) \max(|g_{ij}| - \eta, 0)\mathbb{I}(i\neq j) + g_{ij}\mathbb{I}(i=j),
\]
where \(\sign(\cdot)\) refers to the sign function, and \(\mathbb{I}(\cdot)\) refers to the indicator function. For \(\Gamma \text{ step}\), we have
\[
\begin{aligned}
\Gamma^{(q+1)} & = \argmin_{\Gamma\succeq \epsilon I} \ell(\Sigma^{(q+1)}, \Gamma, \Lambda^{(q)})\\
& = \argmin_{\Gamma\succeq \epsilon I} \frac{1}{2}\lVert \Sigma^{(q+1)} - S^l\rVert^2_F + \frac{1}{2}\lVert \Sigma^{(q+1)} - S^u\rVert^2_F + \lambda \lVert \Sigma^{(q+1)}\rVert \\
& - \langle \Lambda^{(q)}, \Sigma^{(q+1)} - \Gamma \rangle + \frac{1}{2\beta} \lVert \Sigma^{(q+1)} - \Gamma\rVert^2_F \\
& = \argmin_{\Gamma\succeq \epsilon I} \langle\Lambda^{(q)}, \Gamma \rangle + \frac{1}{2\beta} \lVert \Sigma^{(q+1)} - \Gamma\rVert^2_F\\
& = \argmin_{\Gamma\succeq \epsilon I} \lVert \Gamma - (\Sigma^{(q+1)} - \beta \Lambda^{(q)}) \rVert^2_F\\
& = (\Sigma^{(q+1)} - \beta \Lambda^{(q)})_+,
\end{aligned}
\]
where \((\cdot)_+\) refers to the operator that projects a matrix onto the convex cone \(\{\Gamma \succeq \epsilon I\}\). For example, for any matrix \(G\), whose eigenvalue decomposition is given by \(\sum_{ i=1}^{p} \lambda_i \nu^\top_i \nu_i\), where \(\lambda_1,\lambda_2,\cdots,\lambda_p\) are the eigenvalues corresponding to the eigenvectors \(\nu_1,\nu_2,\cdots,\nu_p\).The projection operation \((G)_+\) can be expressed as \(\sum_{ i=1}^{p} (\lambda_i \vee \epsilon) \nu^\top_i \nu_i\). For \(\Lambda \text{ step}\), we can easily have
\[
\begin{aligned}
\Lambda^{(q+1)} = \Lambda^{(q)} - \frac{1}{\beta}(\Sigma^{(q+1)} - \Gamma^{(q+1)}). 
\end{aligned}
\]

We summarize the three-step parameter update process in the following algorithm.d
\begin{algorithm}[H]
\caption{Alternating direction solver for IST estimator with positive definite constraint} \label{ADMMalgorithm}
\begin{algorithmic}
\REQUIRE ~~\\
Input \(\Gamma^{(0)}, \Lambda^{(0)}, \beta\).
\ENSURE ~~\\	
For the \((q+1)\)-th iteration\\
\qquad 1. Solve \(\Sigma^{(q+1)}=\frac{1}{2\beta+1} \gS(\beta(S^l + S^u + \Lambda^{(q)}) + \Gamma^{(q)}, \beta \lambda)\);\\
\qquad 2. Solve \(\Gamma^{(q+1)} = (\Sigma^{(q+1)} - \beta \Lambda^{(q)})_+\);\\
\qquad 3. Solve \(\Lambda^{(q+1)} = \Lambda^{(q)} - \frac{1}{\beta}(\Sigma^{(q+1)} - \Gamma^{(q+1)})\);\\
Repeat the above cycle till convergence;\\
\textbf{Output}: IST estimator.\\
\end{algorithmic}
\end{algorithm}

\section{Theoretical Properties} \label{sec3}
In this section, we discuss the convergence of Algorithm \ref{ADMMalgorithm} and establish the non-asymptotic statistical theory for the IST estimator under different tail conditions.

\subsection{Convergence Analysis of Algorithm \ref{ADMMalgorithm}}
Let \(\Sigma^* = (\sigma^*_{ij})_{1\leq i,j\leq p}\) and \(\Gamma^*\) be the optimal solutions of optimization problem (\ref{consensusIST}), and \(\Lambda^*\) be the corresponding optimal dual variable. Our goal is to prove that the sequence \((\Sigma^{(q)}, \Gamma^{(q)}, \Lambda^{(q)})\) generated by Algorithm \ref{ADMMalgorithm} converges to \((\Sigma^*, \Gamma^*, \Lambda^*)\). First, we define some necessary notations. Let matrix \(D \in \mathbb{R}^{2p\times 2p}\) be
\[
D = \left(
\begin{array}{cc}
\beta I & 0\\
0 & \frac{2}{\beta} I\\
\end{array}
\right),
\]
where \(I\in \mathbb{R}^{p\times p}\) is the identity matrix. Based on \(D\), we define the norm \(\lVert \cdot \rVert^2_D\) as \(\lVert V \rVert^2_D = \langle V, DV\rangle\), and the inner product \(\langle \cdot, \cdot \rangle_D\) as \(\langle V, U \rangle_D = \langle V, DU\rangle\), where \(V\) and \(U\) are appropriately dimensioned matrices. Before presenting the convergence result, we first have the following auxiliary lemma.

\begin{lemma}\label{auxiliarylemma}
Let \(U^* = (\Lambda^*, \Sigma^*)^\top, U^{(q)} = (\Lambda^{(q)}, \Sigma^{(q)})^\top\). Then, we have that the sequence \((\Sigma^{(q)}, \Gamma^{(q)}, \Lambda^{(q)})\) generated by Algorithm \ref{ADMMalgorithm} satisfies
\[
\lVert U^{(q)} - U^*\rVert^2_D - \lVert U^{(q+1)} - U^*\rVert^2_D \geq \lVert U^{(q)} - U^{(q+1)} \rVert^2_D.
\]
\end{lemma}

In Lemma \ref{auxiliarylemma}, \( \| U^{(q)} - U^* \|_D^2 - \| U^{(q+1)} - U^* \|_D^2 \) represents the reduction in the distance between the iterate \( U^{(q)} \) at the \( q \)-th iteration and the optimal solution \( U^* \), compared to the distance at the \( (q+1) \)-th iteration. \( \| U^{(q)} - U^{(q+1)} \|_D^2 \) represents the ``jump'' size between two consecutive iterations, that is, the distance between two consecutive iterates. This indicates that, in each iteration, the reduction in the squared distance to the optimal solution \( U^* \) is greater than or equal to the squared distance between consecutive iterates, meaning that each step of the optimization process effectively brings the solution closer to the optimal one. Using Lemma \ref{auxiliarylemma}, we immediately obtain the following convergence result.

\begin{theorem} \label{maintheorem}
The sequence \( (\Sigma^{(q)}, \Gamma^{(q)}, \Lambda^{(q)}) \) generated by Algorithm \ref{ADMMalgorithm}, initialized from any point, satisfies 
\begin{itemize}
\item \( \| U^{(q)} - U^{(q+1)} \|_D \to 0 \),
\item \((U^{(q)})\) lies in a compact region,
\item \( \|U^{(q)} - U^* \|_D^2 \) is monotonically non-increasing,
\end{itemize}
and thus can converge to the optimal solution \( (\Sigma^*, \Gamma^*, \Lambda^*) \) of optimization problem (\ref{consensusIST}).
\end{theorem}

\subsection{Non-asymptotic analysis of the IST estimator}
In this section, we perform a non-asymptotic analysis of the IST estimator under exponential-tail and polynomial-tail conditions. Before doing so, we define some useful notation. Let \(\Sigma^0 = (\sigma^0_{ij})_{1\leq i,j\leq p}\) be the true covariance matrix corresponding to the observations \(Y = (Y_{ij})_{1\leq i\leq n, 1\leq j\leq p}\), and define the support set of \(\Sigma^0\) as \(\gA = \{(i,j): \sigma^0_{ij} \neq 0, i \neq j\}\), with cardinality \(\alpha = |\gA|\), \(\gA^c \) is the complement of \(\gA\). Define \(\sigma_{\max}= \max_{ii} \sigma^0_{ii}\) as the maximum true variance in \(\Sigma^0\). Additionally, for matrix \(A = (a_{ij})_{1\leq i,j\leq p}\), let \(A_\gA = (a_{ij}\mathbb{I}((i,j)\in \gA))_{1\leq i,j\leq p}\), where \(\mathbb{I}(\cdot)\) is the indicator function. We first have the following theorem, which serves as a bridge for our non-asymptotic analysis.

\begin{theorem}\label{nasymptoticth1}
The Frobenius norm of the difference between the optimal solution \(\Sigma^*\) and the true covariance matrix \(\Sigma^0\) is bounded with high probability, i.e.,
\[
\PP\left(\lVert\Sigma^* - \Sigma^0 \rVert_F \leq 5\lambda (\alpha + p)^{\frac{1}{2}}\right) \geq 1 - \PP\left(\max_{i,j} (|s^l_{ij} - \sigma^0_{ij}| \vee |s^u_{ij} - \sigma^0_{ij}|) > \lambda\right).
\]
\end{theorem}

Theorem \ref{nasymptoticth1} indicates that the non-asymptotic analysis of the optimal solution can be translated into an analysis of the difference between the empirical covariance matrix and the true covariance matrix. This is why the analysis is conducted under two different tail conditions. The following discuss each case separately.

C1 (Exponential-tail condition): There exists \(\eta > 0\) and \(K_2 > 0\), such that for any \(|t| \leq \eta\) and \(1\leq i \leq n, 1\leq j \leq p\):
\[
\E\left(\exp(t(Y^l_{ij})^2)\right) \vee \E\left(\exp(t(Y^u_{ij})^2)\right) \leq K_2.
\]

From the exponential tail condition C1, we can easily deduce that there exists a constant \(K_1, K_3 >0\) such that for any \(1\leq k \leq p\):
\[
(1+\E(t|Y^l_{ij}|)) \vee (1 + \E(t|Y^u_{ij})|) \leq K_1,\  (1+t\E(Y^l_{ik}) \E(Y^l_{ij})) \vee (1+t\E(Y^u_{ik}) \E(Y^u_{ij})) \leq K_3.
\]

The requirement \(\sup_{|t|\le \eta}\mathbb{E}\big[\exp\{t(Y^{\square}_{ij})^2\}\big]\le K_2\) (for \(\square\in\{l,u\}\)) constitutes a uniform \(\psi_2\)-type control on the interval endpoints, ensuring light tails and enabling exponential concentration to transfer from bounds \((Y^l,Y^u)\) to the latent point observations \(X_{ij}\in[Y^l_{ij},Y^u_{ij}]\); C1 immediately yields finite low-order moments and bounded products of expectations (captured by constants \(K_1,K_3\)), which feed into a high-probability, entrywise control of \(\|S-\Sigma^0\|_\infty\); these constants influence only multiplicative factors and do not alter the canonical rate \(\sqrt{\log p/n}\).

\begin{theorem}\label{nasymptoticth2}
Assume that \(n \geq \log p\). For any \(M >0\), under condition C1, by further setting regularization parameter \(\lambda = \tau^2_0 + \tau_1\), where 
\[
\tau_0 = c_0 (\frac{\log p}{n})^{1/2}, c_0= \frac{1}{\log p} (M+1 + \frac{1}{2} K_1n K_2e),
\]
and
\[
\tau_1 = c_1(\frac{\log p}{n})^{1/2}, \ c_1 = \frac{2}{n\eta} \log^{Q_4}_p
\]
with 
\[
Q_4 = \exp \left(\frac{\sigma_{\max}}{2} \frac{1}{2} \eta (\frac{\log p}{n})^{1/2}\right) nK_3K_2 \left(4+ \sigma^2_{\max} \frac{\eta^2}{4} \frac{\log p}{n} + (2K_3 -1)\frac{1}{2} \eta (\frac{\log p}{n})^{1/2}\right),
\]
we then have
\[
\PP\left(\lVert\Sigma^* - \Sigma^0 \rVert_F \leq 5\lambda (\alpha + p)^{\frac{1}{2}}\right) \geq 1- 6p^{-M}.
\]
\end{theorem}

The bound \(\|\Sigma^*-\Sigma^0\|_F \le 5\lambda(\alpha+p)^{1/2}\) with \(\lambda=\tau_0^2+\tau_1\asymp \sqrt{\log p/n}\) matches classical high-dimensional rates for point data under sub-Gaussian tails. The factor \((\alpha+p)^{1/2}\) reflects decomposability of the off-diagonal \(\ell_1\) penalty and the unavoidable estimation of \(p\) variances; the split \(\tau_0^2+\tau_1\) separates variance-like and bias/modeling contributions without changing the leading order.

The probability \(1-6p^{-M}\) arises from uniform (entrywise) control and union bounds across matrix indices and endpoint events, while the PSD floor \(\varepsilon I\) stabilizes estimation when eigenvalues approach zero. Constants depending on \((\eta,K_2,K_1,K_3,\sigma_{\max})\) enter only via multiplicative factors (e.g., \(Q_4\)), leaving the dominant \(\sqrt{\log p/n}\) rate intact provided \(n\ge \log p\).

C2 (Polynomial-tail condition): There exists \(K_4 > 0\), for any \(\gamma >0, \epsilon>0\) and \(1\leq i \leq n, 1\leq j \leq p\):
\[
\E(|Y^l_{ij}|^{4(1+\gamma+\epsilon)}) \vee \E(|Y^u_{ij}|^{4(1+\gamma+\epsilon)}) \leq K_4.
\]

The moment requirement \(\mathbb{E}\big|Y^{\square}_{ij}\big|^{4(1+\gamma+\epsilon)}\le K_4\) accommodates \emph{heavy-tailed endpoints} while retaining sufficient integrability to control deviations of \(S\); because moment-based inequalities yield polynomial (rather than exponential) tails, one must restrict dimensional growth (e.g., \(p\le c n^\gamma\)) or employ mild robustification (truncation/Huberization) to keep union bounds effective.

\begin{theorem}\label{nasymptoticth3}
Under Condition 2, assume \(p \leq cn^{\gamma}\) for some \(c>0\). For \(M >0\), select the regularization parameter 
\[
\lambda = 8(K_4 + 1)(M + 1)\left( \frac{\log p}{n} \right) + 8(K_4+1)(M+2) \left(\frac{\log p}{n}\right)^{1/2}.
\]

Then, with a probability of at least \(1-O(p^{-M}) - 6K_4p(\log n)^{2(1+\gamma+\tau)}n^{-(\gamma+\tau)}\), we have
\[
\lVert\Sigma^* - \Sigma^0 \rVert_F \leq 5\lambda (\alpha + p)^{\frac{1}{2}}.
\]
\end{theorem}

With \(\lambda=8(K_4+1)\big\{(M+1)\tfrac{\log p}{n}+(M+2)\sqrt{\tfrac{\log p}{n}}\big\}\), the estimator attains the same order \(\|\Sigma^*-\Sigma^0\|_F=O_\mathbb{P}\!\big(\sqrt{\tfrac{\log p}{n}}\,(\alpha+p)^{1/2}\big)\) as in the light-tailed, point-data literature. The extra \(\log p/n\) term buffers heavy tails; in moderate samples the \(\sqrt{\log p/n}\) component is typically dominant.

The confidence level includes an additional polynomial term \(6K_4\,p(\log n)^{2(1+\gamma+\tau)}n^{-(\gamma+\tau)}\), reflecting moment-based concentration and the growth constraint \(p\le c n^\gamma\). Practically, choosing \(\lambda=C\sqrt{\log p/n}\) together with a robust construction of \(S\) (e.g., truncation at high quantiles) yields behavior consistent with the theorem while improving constants in heavy-tail regimes.

\section{Simulation Studies} \label{sec4}
\subsection{Baseline Simulation: Covariance Matrix Estimation}

In this section, we conduct simulation studies to evaluate the effectiveness of the proposed method. We first describe the simulation setup. Specifically, we consider three different DGPs. DGP1: The lower bound of the interval is generated first, and a fixed positive constant is added to obtain the upper bound; DGP2: The center of the interval is generated first, and a fixed positive constant is added and subtracted to obtain the upper and lower bounds, respectively; DGP3: The center of the interval is generated first, and then a non-negative independent random variable is added and subtracted to obtain the upper and lower bounds, respectively. For DGP1 and DGP2, we consider four distinct constants—0.5, 1, 3, and 5—to investigate the impact of interval length on estimation efficiency. For DGP3, we examine four non-negative random variables: lognormal, beta, gamma, and exponential. Since the constant does not affect the covariance, the upper and lower bounds in DGP1 and DGP2 share the same covariance matrix. Similarly, for DGP3, due to independence, the upper and lower bounds also share the same covariance matrix, which, compared to the original center's covariance matrix, only differs by the addition of the variance of the non-negative random variables on the diagonal elements.

We generate samples for the lower bounds in DGP1 and the centers in DGP2 and DGP3 from a multivariate normal distribution with zero mean and covariance matrix \(\Sigma^0 = (\sigma^0_{ij})_{1\leq i,j\leq p}\). We investigate three distinct covariance structures, detailed as follows.
\begin{itemize}
\item Moving average covariance structure (MA(1)). The covariance matrix \(\Sigma^0\) is derived from an MA(1) process, defined as: \(\sigma^0_{ij} = \rho^{|i-j|} \cdot \mathbb{I}(|i-j| \leq 1)\), where \(\rho\) controls the correlation strength, and the indicator function restricts non-zero covariances to adjacent indices.

\item Autoregressive covariance structure (AR(1)). The covariance matrix is given by:\(\sigma^0_{ij} = \rho^{|i-j|}\), where \(\rho\) determines the exponential decay of correlations with increasing distance between indices.

\item Long-range dependence structure (LR). The covariance matrix is defined as
\[
\sigma^0_{ij} = \frac{1}{2} \left[ (|i - j| + 1)^{2H} - 2|i - j|^{2H} + (|i - j| - 1)^{2H} \right], \quad 1 \leq i,j \leq p,
\]
where \(H \in [0.5, 1]\) is the Hurst parameter. When \(H = 0.5\), the process reduces to white noise. As \(H\) increases, the process exhibits stronger long-range dependence. In practical applications, values of \(H\) up to 0.9 are frequently observed, for instance in modeling Internet network traffic.
\end{itemize}

We explore parameter combinations where the dimension \(p\) varies with the sample size \(n\). Specifically, for \(n = 100\), we set \(p \in \{100, 120, 150\}\); for \(n = 120\), \(p \in \{120, 150, 180\}\); and for \(n = 150\), \(p \in \{150, 180, 200\}\). This design ensures \(p\) scales with \(n\) while testing both moderate and high-dimensional regimes, with an upper limit of \(p = 200\) to maintain computational feasibility. For the three types of covariance matrices, we consider three different correlation coefficients, \(\rho = 0.1, 0.5, 0.9\). For the Hurst parameter in DGP3, we examine three values: \(H = 0.5, 0.7, 0.9\). Let \(\widehat{\Sigma}\) denote the estimated covariance matrix. In all simulation studies, we use the Frobenius norm \(\lVert \widehat{\Sigma} - \Sigma^0 \rVert_F\) and the spectral norm \(\lVert \widehat{\Sigma} - \Sigma^0 \rVert_2\) as performance metrics to evaluate the quality of the estimators.

Given the multitude of results arising from diverse experimental configurations, we report averaged values where appropriate to effectively convey the key findings. For instance, the entry in the first cell of Table \ref{tab:combined1} represents the average Frobenius norm of the estimated covariance matrices across different correlation coefficients. Tables \ref{tab:combined1}–\ref{tab:combined3} present the results for DGP1 through DGP3.

\begin{table}[H]
\centering
\caption{Experimental Results for DGP1.}
\label{tab:combined1}
\setlength{\tabcolsep}{1.2mm}{
\begin{tabular}{c|ccc|ccc|ccc}
\hline
\multirow{2}{*}{Covariance} & \multicolumn{3}{c|}{$n=100$} & \multicolumn{3}{c|}{$n=120$} & \multicolumn{3}{c}{$n=150$} \\
\cline{2-10}
& $p=100$ & $p=120$ & $p=150$ & $p=120$ & $p=150$ & $p=180$ & $p=150$ & $p=180$ & $p=200$ \\
\hline
\multicolumn{10}{c}{Frobenius norm (average)} \\
\hline
MA(1) & 7.486 & 8.187 & 9.125 & 8.152 & 9.104 & 9.985 & 9.023 & 9.912 & 10.466 \\
AR(1) & 12.908 & 14.187 & 15.916 & 14.140 & 15.866 & 17.418 & 15.813 & 17.361 & 18.326 \\
LR & 16.328 & 18.792 & 22.325 & 18.713 & 22.230 & 25.598 & 22.166 & 25.531 & 27.718 \\
\hline
\multicolumn{10}{c}{spectral norm (average)} \\
\hline
MA(1) & 1.311 & 1.320 & 1.323 & 1.280 & 1.302 & 1.300 & 1.259 & 1.253 & 1.290 \\
AR(1) & 6.457 & 6.540 & 6.642 & 6.543 & 6.644 & 6.680 & 6.621 & 6.668 & 6.706 \\
LR & 14.900 & 17.141 & 20.359 & 17.115 & 20.320 & 23.410 & 20.310 & 23.385 & 25.358 \\
\hline
\end{tabular}
}
\end{table}

From Table \ref{tab:combined1}, it is evident that the estimator under the MA(1) structure demonstrates superior performance across all experimental configurations, with Frobenius norms ranging from 7.486 to 10.466—substantially lower than those for the other two structures. The estimator under the AR(1) structure exhibits intermediate performance, with norms between 12.908 and 18.326, approximately 1.7 to 1.8 times those of MA(1). In contrast, the estimator under the LR structure performs the poorest, with Frobenius norms ranging from 16.328 to 27.718, roughly 2.2 to 2.6 times those of MA(1). The spectral norm results further support these findings. Under the MA(1) structure, the estimator maintains a low spectral norm between 1.253 and 1.323, indicating excellent numerical stability. The spectral norms for AR(1) range from 6.457 to 6.706, about five times higher than those for MA(1). The LR structure yields the highest spectral norms, ranging from 14.900 to 25.358, approximately 10 to 20 times those of MA(1). Moreover, as the dimension \(p\) increases, estimation errors rise across all covariance structures, though the growth is slowest for MA(1). Specifically, when \(p\) increases from 100 to 200, the Frobenius norm for MA(1) increases by only about 40\%, compared to roughly 70\% for LR, highlighting the greater challenges of estimating LR structures in high-dimensional settings. By contrast, increasing the sample size \(n\) has a relatively modest effect on estimation accuracy; when \(n\) increases from 100 to 150, the improvements across all structures are limited. This suggests that dimensionality has a stronger impact than sample size on covariance matrix estimation performance.

\begin{table}[H]
\centering
\caption{Experimental Results for DGP2.}
\label{tab:combined2}
\setlength{\tabcolsep}{1.2mm}{
\begin{tabular}{c|ccc|ccc|ccc}
\hline
\multirow{2}{*}{Covariance} & \multicolumn{3}{c|}{$n=100$} & \multicolumn{3}{c|}{$n=120$} & \multicolumn{3}{c}{$n=150$} \\
\cline{2-10}
& $p=100$ & $p=120$ & $p=150$ & $p=120$ & $p=150$ & $p=180$ & $p=150$ & $p=180$ & $p=200$ \\
\hline
\multicolumn{10}{c}{Frobenius norm (average)} \\
\hline
MA(1) & 7.444 & 8.167 & 9.184 & 8.128 & 9.121 & 9.981 & 9.052 & 9.896 & 10.444 \\
AR(1) & 12.887 & 14.231 & 15.910 & 14.115 & 15.843 & 17.437 & 15.809 & 17.372 & 18.343 \\
LR & 16.340 & 18.792 & 22.308 & 18.746 & 22.226 & 25.597 & 22.191 & 25.534 & 27.708 \\
\hline
\multicolumn{10}{c}{spectral norm (average)} \\
\hline
MA(1) & 1.287 & 1.306 & 1.322 & 1.285 & 1.290 & 1.315 & 1.269 & 1.257 & 1.281 \\
AR(1) & 6.440 & 6.569 & 6.647 & 6.537 & 6.625 & 6.694 & 6.611 & 6.675 & 6.718 \\
LR & 14.869 & 17.133 & 20.354 & 17.142 & 20.309 & 23.435 & 20.306 & 23.391 & 25.380 \\
\hline
\end{tabular}
}
\end{table}

Table \ref{tab:combined2} reports the average Frobenius and spectral norms of the estimation errors under three covariance structures—MA(1), AR(1), and LR—across different dimensional settings. Several clear patterns emerge. First, for all three covariance structures, the Frobenius norm increases steadily as the dimensionality 
\(p\) grows, reflecting the natural accumulation of estimation error in higher-dimensional settings. Specifically, for the MA(1) structure, the Frobenius norm rises from 7.444 at \(p=100\) to 10.444 at \(p=200\), indicating relatively moderate growth. In contrast, the AR(1) structure exhibits higher estimation errors overall, increasing from 12.887 to 18.343 over the same range. The LR structure shows the most pronounced growth, from 16.340 to 27.708, underscoring the greater difficulty of estimating long-range dependent covariance matrices in high-dimensional settings. Second, when examining the spectral norm, we observe that the MA(1) structure remains stable and low across dimensions, with values clustered around 1.25–1.32. This suggests that the largest singular value errors are relatively contained for short-memory processes. By contrast, the AR(1) structure shows moderate but consistent growth, from 6.440 at \(p=100\) to 6.718 at \(p=200\). Finally, the LR structure exhibits the largest spectral norm errors, escalating from 14.869 to 25.380 as \(p\) increases, mirroring the Frobenius norm pattern.

\begin{table}[H]
\centering
\caption{Experimental Results for DGP3.}
\label{tab:combined3}
\setlength{\tabcolsep}{1.2mm}{
\begin{tabular}{c|ccc|ccc|ccc}
\hline
\multirow{2}{*}{Covariance} & \multicolumn{3}{c|}{$n=100$} & \multicolumn{3}{c|}{$n=120$} & \multicolumn{3}{c}{$n=150$} \\
\cline{2-10}
& $p=100$ & $p=120$ & $p=150$ & $p=120$ & $p=150$ & $p=180$ & $p=150$ & $p=180$ & $p=200$ \\
\hline
\multicolumn{10}{c}{Frobenius norm (average)} \\
\hline
MA(1) & 17.924 & 13.304 & 16.366 & 13.891 & 18.053 & 56.409 & 45.830 & 47.441 & 48.709 \\
AR(1) & 22.418 & 18.747 & 22.315 & 19.278 & 23.889 & 62.149 & 51.171 & 53.261 & 54.873 \\
LR & 25.603 & 23.216 & 28.499 & 23.720 & 29.947 & 68.971 & 56.433 & 60.179 & 62.801 \\
\hline
\multicolumn{10}{c}{spectral norm (average)} \\
\hline
MA(1) & 11.051 & 4.252 & 4.967 & 4.336 & 8.007 & 46.882 & 37.333 & 37.528 & 37.443 \\
AR(1) & 14.709 & 8.483 & 9.097 & 8.608 & 11.858 & 50.770 & 41.286 & 41.451 & 41.392 \\
LR & 21.067 & 18.706 & 22.309 & 18.814 & 24.366 & 63.348 & 51.485 & 54.035 & 55.375 \\
\hline
\end{tabular}
}
\end{table}

Table \ref{tab:combined3} presents the estimation errors for DGP3 under different covariance structures, sample sizes, and dimensions. Across all configurations, the MA(1) structure consistently achieves the lowest errors, indicating that short-memory dependence is easier to estimate. The AR(1) structure shows moderate errors, whereas the LR structure exhibits the largest estimation errors, especially in high-dimensional regimes. For example, the Frobenius norm for LR increases from 25.603 at \((n,p)=(100,100)\) to 62.801 at \((n,p)=(150,200)\), compared with 17.924 to 48.709 for MA(1). The spectral norm results display similar patterns. At lower dimensions, spectral norms remain moderate, but they increase sharply for larger \(p\), particularly for LR, which reaches values above 55. This indicates that estimation becomes increasingly unstable in the presence of long-range dependence and large dimensionality. Finally, increasing the sample size \(n\) from 100 to 150 has a relatively minor effect compared to increasing \(p\). This suggests that, for DGP3, dimensionality is the dominant factor influencing estimation performance, especially for long-range dependent covariance structures.

A comparison of Tables \ref{tab:combined1}–\ref{tab:combined3} shows that the estimator exhibits stable and consistent performance across all three DGPs, with clear variations depending on the underlying covariance structure and dimensionality. For MA(1) structures, the estimator consistently achieves the lowest Frobenius and spectral norms, demonstrating high estimation accuracy and numerical stability even as the dimension increases. Under AR(1) structures, the performance is moderate—errors are larger than those for MA(1) but remain well controlled, reflecting the relatively simple dependence structure. In contrast, LR structures present the greatest estimation challenges, with both norms substantially higher, particularly in high-dimensional scenarios. Moreover, dimensionality has a much stronger effect on estimation accuracy than sample size, especially for long-range dependent processes. Overall, these findings indicate that the estimator performs well for short- and moderate-range dependence but becomes less accurate when facing strong long-range dependence and large dimensions, which aligns with theoretical expectations. Given the large number of combinations, we randomly select the case under DGP1 with an AR(1) covariance structure and \((n,p)=(120,120)\) for visualization, as shown in Figure \ref{covarianceanalysis}. From the figure, we observe that the estimated covariance matrix closely resembles the true one, with only minor differences. In particular, the eigenvalues of the estimated matrix closely follow those of the true matrix, lying almost on the same curve. Moreover, the estimated variances fluctuate only slightly around the true value of 1. These observations demonstrate the favorable properties and accuracy of the estimator.
\begin{figure}[H]
\centering
\includegraphics[scale=0.38]{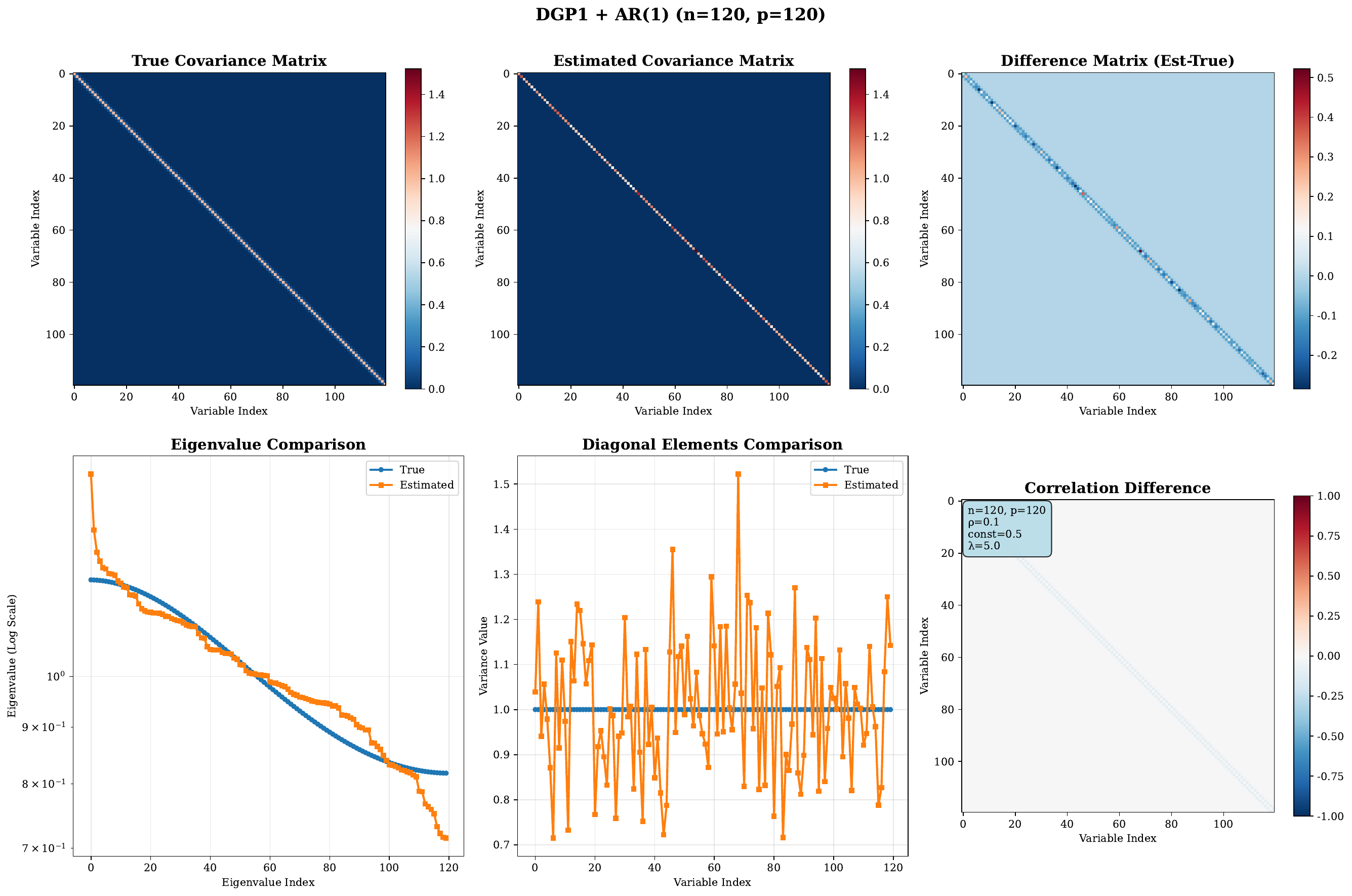}
\caption{Estimation results for DGP1 with an AR(1) covariance structure and \((n,p)=(120,120)\). In the first row, from left to right, are the true covariance matrix, the estimated covariance matrix, and the difference between the two. In the second row, from left to right, are the eigenvalue curves, the variance scatter plot, and the interpolated correlation matrix.}
\label{covarianceanalysis}
\end{figure}

Based on the simulation results, we next examine how the covariance structure, sample size, dimensionality, and constant parameter influence the estimation accuracy of the covariance matrix. The corresponding outcomes are illustrated in Figure \ref{spthreesy}, highlighting the sensitivity of the estimator to these factors. 

\begin{figure}[H]
\centering 
\begin{minipage}[t]{1\textwidth}
\centering
\includegraphics[scale=0.35]{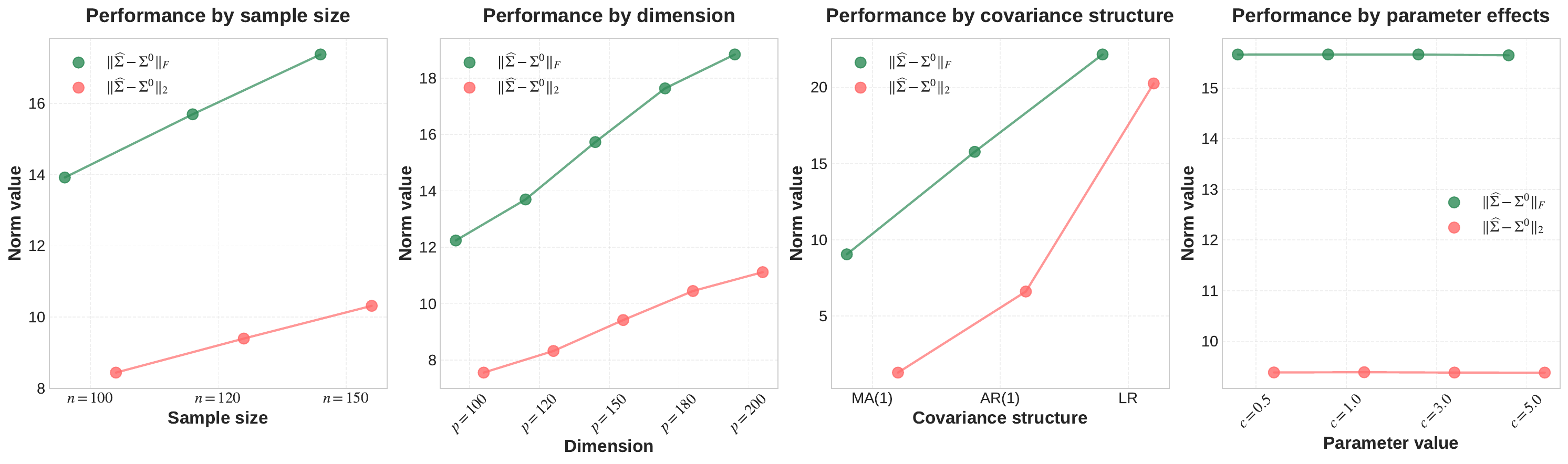}
\end{minipage}

\begin{minipage}[t]{1\textwidth}
\centering
\includegraphics[scale=0.35]{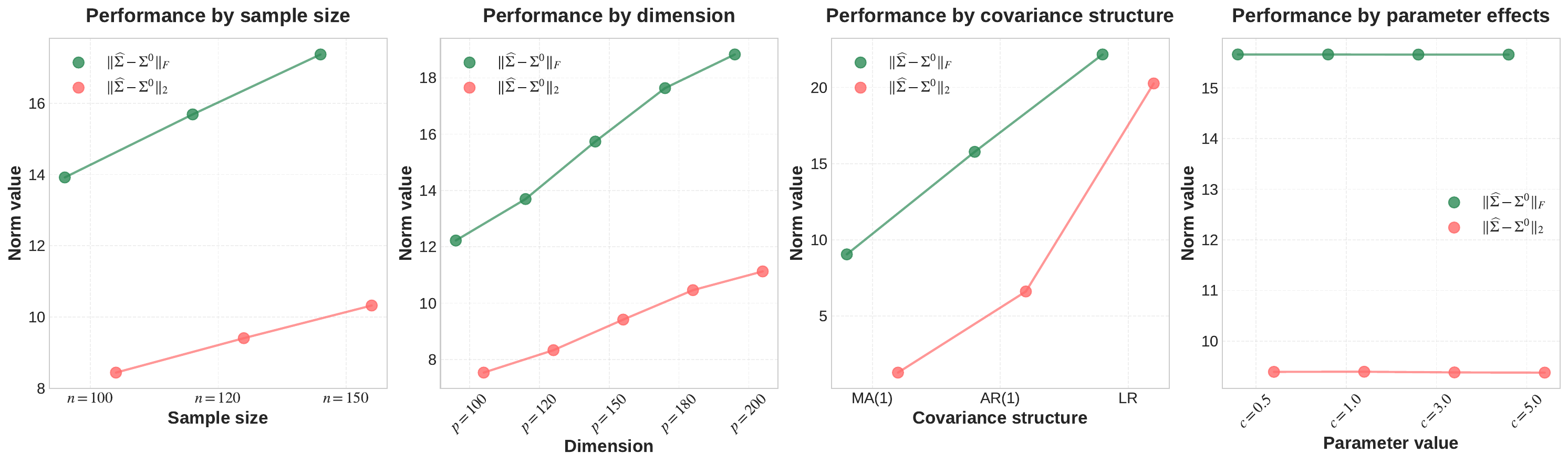}
\end{minipage}

\begin{minipage}[t]{1\textwidth}
\centering
\includegraphics[scale=0.35]{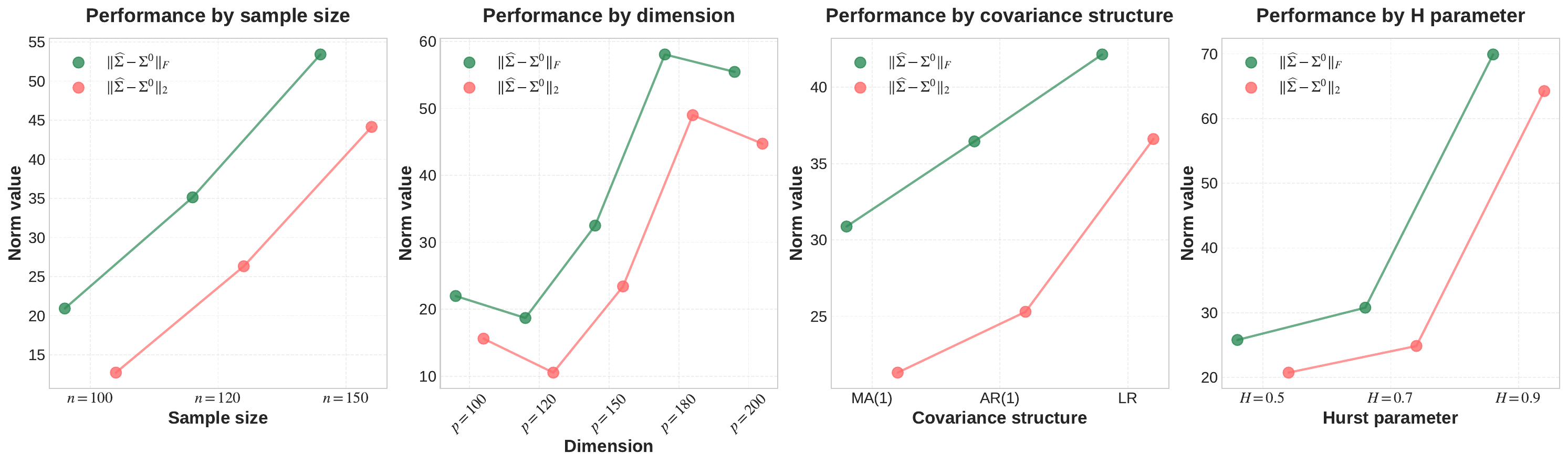} 
\end{minipage}
\caption{The effects of covariance structure, sample size, dimensionality, and parameter on estimation efficiency; the three rows correspond to DGP1, DGP2, and DGP3, respectively.}\label{spthreesy} 
\end{figure}

The patterns observed in Figure \ref{spthreesy} are consistent with the conclusions drawn from Tables \ref{tab:combined1}–\ref{tab:combined3}, and the estimator exhibits similar behavior across all three DGPs. Specifically, as the dimensionality increases, the estimation becomes more challenging, resulting in higher values for both norms. In terms of estimation difficulty, the LR structure is markedly more challenging than AR(1), which in turn is more difficult than MA(1). Furthermore, for DGP1 and DGP2, the length of the interval appears to have little effect on estimation efficiency, which aligns with our theoretical results: only the assumption of the same dependence structure between the upper and lower bounds is required, without imposing restrictions on the interval length.

\subsection{Guided Simulation: Toward Empirical Application}

In Section \ref{sec5}, we apply the proposed method to intraday high-frequency financial data, estimating covariance matrices based on interval-valued observations. To validate the empirical application, this section presents a guided simulation study in which synthetic data are generated to align with the continuous-time Itô semimartingale framework, a standard model for high-frequency financial data.

In this simulation, we first generate second-by-second return processes for a total of $p$ financial assets. The synthetic high-frequency data are drawn from a continuous-time diffusion process with a constant covariance structure. Specifically, we simulate a $p$-dimensional Itô process
\begin{align*}
dY_t = \Sigma^{1/2}dW_t, \quad t \in [0,1],
\end{align*}
where $W_t$ is a standard $p$-dimensional Brownian motion and $\Sigma$ denotes the instantaneous covariance matrix. The covariance $\Sigma$ is specified as a Toeplitz matrix with entries $\rho^{|i-j|}$. Marginal volatilities are set to unity, i.e., $\sigma_i = 1$ for all $i$. The process is discretized over $n = 23{,}400$ equally spaced time points, corresponding to one trading day of 6.5 hours at one-second intervals. Increments are generated via the Euler discretization scheme as Gaussian draws with covariance $\Sigma , \Delta t$, and their cumulative sums produce the continuous sample paths ${X_t}_{t \in [0,1]}$.

To replicate the practical structure of interval-valued high-frequency data, the simulated trajectories are aggregated into non-overlapping time blocks of 300 seconds (5 minutes) each, resulting in $n_{\text{blocks}} = 78$ intervals per trading day. Within each block, we demean the returns by subtracting the first observation of the block for each asset and then compute the cross-sectional minimum and maximum. This procedure yields the lower and upper bounds of interval-valued observations, which are subsequently used as the observed data for covariance estimation and further analysis.

\begin{table}[H]
\small
\centering
\caption{Frobenius norm errors across dimensions and correlation levels}
\label{tab:all_corr_errors}
\setlength{\tabcolsep}{3.5mm}{
\begin{tabular}{c|cccccc}
\hline
\multirow{2}{*}{Dimension ($p$)} & \multicolumn{6}{c}{Correlation ($\rho$)} \\
\cline{2-7}
& 0.1 & 0.2 & 0.3 & 0.4 & 0.5 & 0.6 \\
\hline
100 & 1.41 & 2.87 & 4.42 & 6.14 & 8.11 & 10.52 \\
200 & 2.00 & 4.07 & 6.27 & 8.70 & 11.51 & 14.94 \\
300 & 2.46 & 4.99 & 7.69 & 10.67 & 14.11 & 18.32 \\
400 & 2.84 & 5.77 & 8.88 & 12.33 & 16.30 & 21.17 \\
500 & 3.17 & 6.45 & 9.93 & 13.78 & 18.23 & 23.68 \\
\hline
\end{tabular}
}
\vspace{2mm}
\begin{minipage}{0.9\textwidth}
\footnotesize
\textit{Notes}: This table reports the average Frobenius norm errors of the estimated covariance matrices across different dimensions ($p$) and correlation levels ($\rho$). 
\end{minipage}
\end{table}

We evaluate the performance of the proposed method across different dimensions, $p \in \{100,200,300,400,500\}$, and correlation levels, $\rho \in \{0.1,0.2,0.3,0.4,0.5,0.6\}$, with the results summarized in Table \ref{tab:all_corr_errors}. As shown, the estimation error, measured by the Frobenius norm, increases with both $p$ and $\rho$, reflecting the expected challenges: higher dimensionality increases parameter uncertainty, while stronger correlations amplify the difficulty of accurately estimating the covariance structure. For a fixed $\rho$, the error grows roughly linearly with $p$, indicating that the proposed estimator scales reasonably well with dimensionality. Across correlation levels, the Frobenius norm error increases more rapidly as $\rho$ rises, highlighting the additional challenge posed by stronger dependence among variables. Nevertheless, the error growth remains smooth and well-behaved across all settings, demonstrating the numerical stability of the proposed method even in high-dimensional, highly correlated scenarios.

\section{Real Data Applications}  \label{sec5}

In this section, we apply the proposed methodology to real data to evaluate its effectiveness. Specifically, we analyze intraday high-frequency financial data for the constituents of the CSI 300 Index in the Chinese stock market. The rapid expansion in the availability of high-frequency data over the past two decades has generated an extensive literature, particularly on volatility estimation. Prominent examples include the two-scale realized volatility estimator \citep{zhang2005tale}, the multi-scale realized volatility estimator \citep{zhang2006efficient}, the realized kernel \citep{barndorff2008designing}, the pre-averaging approach \citep{jacod2009microstructure}, and the quasi-maximum likelihood estimator \citep{xiu2010quasi}. However, most of these methods rely exclusively on the return process and overlook additional information, such as the highest and lowest prices within each sampling interval. In practice, interval-valued high-frequency data are widely available, typically in the form of candlestick data, which record the open, high, low, and close prices for each interval (e.g., 5 minutes). Motivated by this, we explore the estimation of high-dimensional covariance matrices using interval-valued high-frequency financial data.

According to the fundamental theorem of asset pricing, the efficient price process for the $q$-th financial asset follows a continuous-time It\^{o} semi-martingale defined on the filtered probability space $(\Omega, {\cal F}, {\cal F}_t, \mathbb{P})$, aligning with the no-arbitrage condition established by \cite{DelbaenSchachermayer1994}. So that without loss of generality, we posit that the observed price processes $\{Y^{(q)}_{t}\}_{t\geq 0}$ follow the stochastic differential equations
\begin{eqnarray}\label{model}
&&Y^{(q)}_t=Y^{(q)}_0+\int_0^t b^{(q)}_sds+\int_0^t \sigma^{(q)}_s dW^{(q)}_s,
\end{eqnarray}
where $(b^{(q)}_t)_{t\in[0,1]}$ and $(\sigma^{(q)}_t)_{t\in[0,1]}$ are predictable and adapted processes, $W^{(q)}_t$ are standard Brownian motions, and $\rho_t$ is the correlation process with $d<W^{(i)},W^{(j)}>_t=\rho^{ij}_tdt$, where $\rho^{ij}_t\in[-1, 1]$ for $t\geq 0$, almost surely. The time interval is normalized to $[0,1]$ to represent a single trading day. Notably, the increments of high-frequency prices (i.e., returns) exhibit conditional independence, a property stemming from the Itô semimartingale framework. This justifies the use of the proposed method.

For each stock, we collect interval-valued high-frequency data at 5-minute intervals. Our objective is to estimate the high-dimensional volatility matrix $(\Bar{\rho}^{ij}\Bar{\sigma}^{(i)}\Bar{\sigma}^{(j)})_{i,j}$, defined as the average of $(\rho^{ij}\sigma^{(i)}\sigma^{(j)})_{i,j}$. The dataset consists of $298$ individual stocks, and we estimate the volatility matrix for two trading days, October 8 and 9, 2024. This sample period is chosen to capture market dynamics during episodes of extreme turbulence. On October 8, 2024, following the week-long National Day holiday, the Chinese stock market staged a sharp rebound, with the CSI 300 Index surging by nearly 6\% amid record trading volumes. The rally was fueled by renewed investor optimism over potential policy support, particularly for the property sector, although gains moderated later in the day as official announcements provided few concrete new measures. The next day, October 9, 2024, the market reversed dramatically, with the CSI 300 Index falling by more than 7\% in its largest single-day decline since the COVID-19 outbreak. The selloff reflected broad-based investor disappointment at the absence of substantial stimulus announcements, leading to profit-taking and widespread sectoral declines. Taken together, these two trading days illustrate the volatility of market sentiment, with October 8 characterized by optimism-driven dispersion in returns and October 9 by a synchronized downturn marked by stronger cross-stock comovement.

\begin{figure}[H]
\centering
\includegraphics[width=0.9\linewidth]{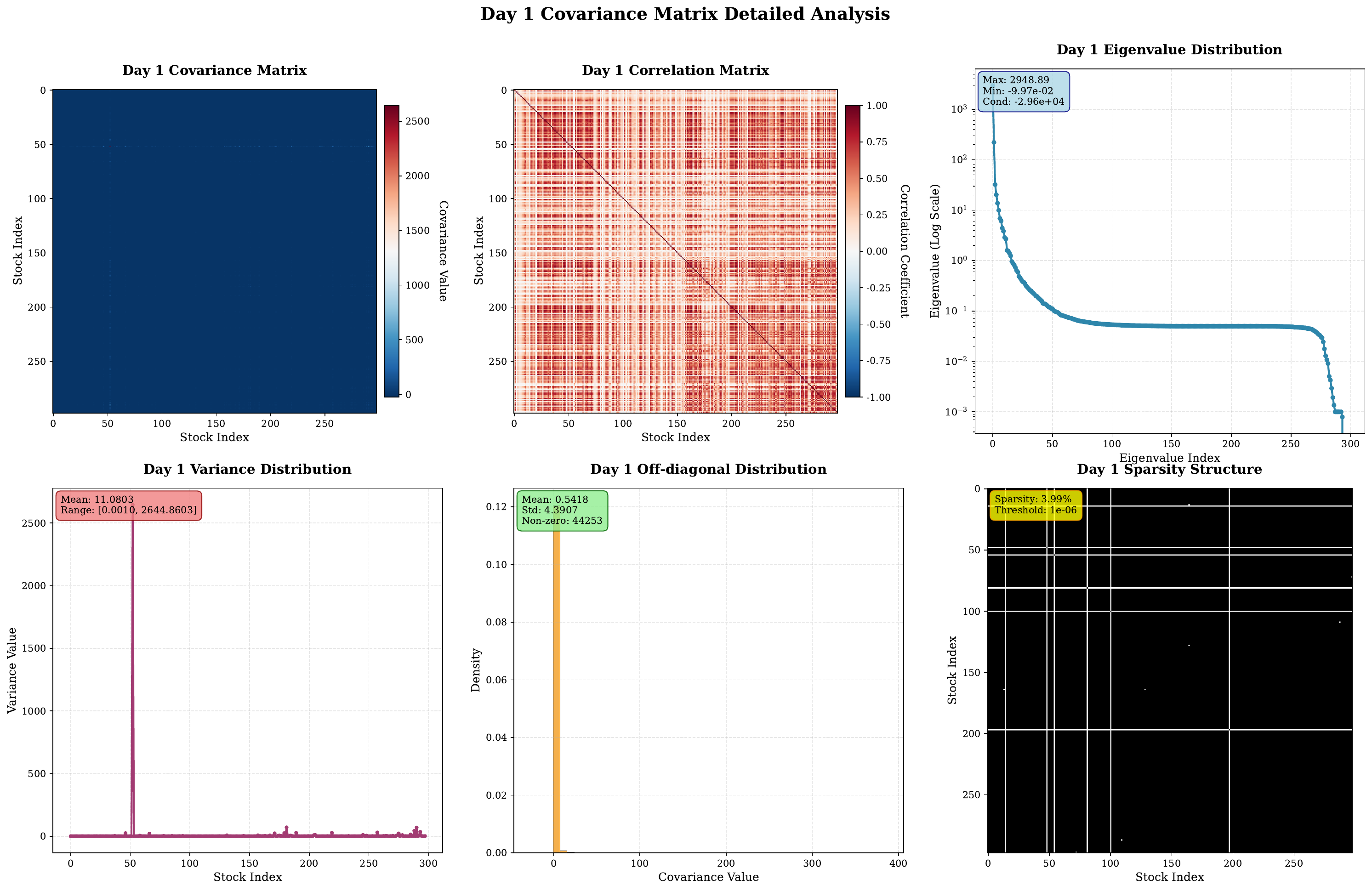}
\caption{The figure presents the estimated covariance matrix as of October 8, 2024.}
\label{fig1}
\end{figure}

The estimation results are reported in Figures \ref{fig1} and \ref{fig2}, where the former corresponds to October 8, 2024 (a booming market day) and the latter to October 9, 2024 (a declining market day). On October 8, 2024, when the Chinese stock market surged following the National Day holiday, the estimated volatility matrix of the CSI 300 constituents exhibited both higher variance dispersion and a pronounced market factor. The largest eigenvalue was considerably larger relative to the subsequent day, indicating that a dominant common component explained a substantial portion of return variation. At the same time, the distribution of variances displayed a wide range, suggesting that certain stocks or sectors experienced particularly large idiosyncratic movements during the rally. Consistent with this, the average pairwise correlation was moderate, implying that while the market factor was strong, cross-sectional heterogeneity in stock behavior contributed to the overall structure of risk. 

\begin{figure}[H]
\centering
\includegraphics[width=0.9\linewidth]{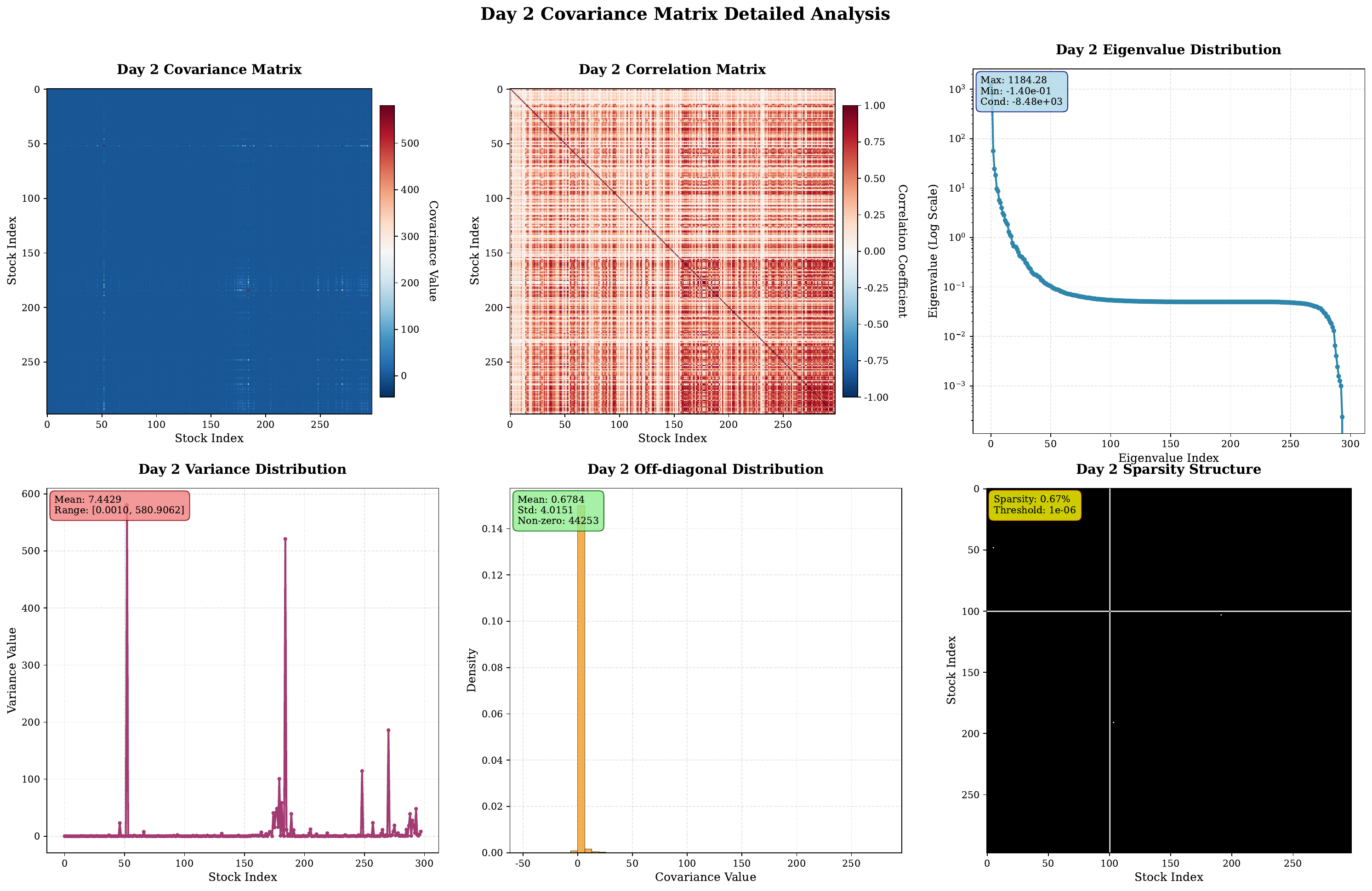}
\caption{The figure presents the estimated covariance matrix as of October 9, 2024.}
\label{fig2}
\end{figure}

By contrast, on October 9, 2024, when the market reversed sharply, the estimated volatility matrix revealed a different pattern. Although the overall level of variances was lower and the leading eigenvalue smaller, the average off-diagonal correlations were markedly higher. This finding suggests that stocks moved more cohesively during the downturn, with comovement intensifying across sectors. Such a pattern is consistent with the well-documented phenomenon that correlations rise in periods of market stress. Together, the results highlight a sharp contrast: October 8 was characterized by optimism-driven dispersion and concentrated sectoral dynamics, whereas October 9 reflected a synchronized market-wide decline with stronger systemic comovement. In summary, by leveraging interval-valued high-frequency data, we obtain economically meaningful conclusions that exploit the additional information contained in the range of the observed price process.

\section{Conclusion} \label{sec6}

In this paper, under the assumption that the upper and lower bounds of intervals exhibited the same dependence structure, we proposed a covariance matrix estimator for high-dimensional interval-valued data that guaranteed positive definiteness. We developed an efficient iterative algorithm based on the ADMM for computing this estimator and theoretically established the convergence of the optimization procedure. We further investigated the statistical properties of the estimator. In particular, under polynomial and exponential tail conditions on the interval upper and lower bounds, we derived high-probability convergence bounds for the estimator. Through extensive simulations across various DGPs and covariance structure assumptions, we demonstrated the estimator's strong finite-sample performance. Additionally, we applied the proposed estimator to high-frequency trading data from the CSI 300 Index, yielding several insightful empirical findings. In future work, we planned to extend the assumption of identical dependence structures between interval upper and lower bounds to the estimation of high-dimensional precision matrices (inverses of covariance matrices) for high-dimensional interval-valued data, thereby enhancing the utilization of the level and volatility information contained in financial datasets.

%

\newpage 
\begin{appendix}
\section{Proofs for Results}\label{appendix}
\subsection{Proof of Lemma \ref{auxiliarylemma}}
\begin{proof}
Given that \(\Sigma^*\) and \(\Gamma^*\) are the optimal solutions of the optimization problem \ref{consensusIST}, and \(\Lambda^*\) is the corresponding optimal dual variable, according to the Karush–Kuhn–Tucker (KKT) conditions, the following conclusion holds:
\begin{align}
\frac{1}{\lambda} (\Lambda^* -2\Sigma^* + S^l+S^u)_{ij} &\in \partial |\Sigma_{ij}^*|, \ i,j =1,2,\cdots,p,\  i\neq j, \label{kkta1}\\
(2\Sigma^* - S^l-S^u)_{ii} - \Lambda^*_{ii} &= 0, \ i = 1,2,\cdots,p,\label{kkta2}\\
\Sigma^* &= \Gamma^*, \label{kkta3}\\
\Gamma^* & \succeq \epsilon I, \label{kkta4}\\\
\langle \Lambda^*, \Gamma^* - \Gamma \rangle &\geq 0, \ \forall \Gamma \succeq \epsilon I.\label{kkta5}\
\end{align}

(\ref{kkta1}) states that for the off-diagonal elements of \(\Sigma^*\), the gradient (subdifferential) of \(|\Sigma_{ij}^*|\) must satisfy a specific relationship with the Lagrange multipliers \(\Lambda^*\) and the matrices \(S^l\), \(S^u\). The subdifferential \(\partial |\Sigma_{ij}^*|\) depends on whether \(\Sigma_{ij}^*\) is zero or non-zero. (\ref{kkta2}) implies that for the diagonal elements of \(\Sigma^*\), the difference between the sum of \(S^l\) and \(S^u\), and the Lagrange multiplier \(\Lambda^*\), must be zero. (\ref{kkta3}) enforces the equality between \(\Sigma^*\) and \(\Gamma^*\), in accordance with the problem's constraints. (\ref{kkta4}) ensures that \(\Gamma^*\) is positive semidefinite, with its eigenvalues greater than or equal to \(\epsilon\). (\ref{kkta5}) expresses the complementary slackness condition, which ensures that if \(\Gamma^*\) is not at its lower bound (\(\epsilon I\)), the corresponding Lagrange multiplier \(\Lambda^*\) must be zero. Conversely, if \(\Gamma^*\) is at its lower bound, \(\Lambda^*\) may be positive.

We then discuss the optimality conditions for \(\Sigma \text{ step}\) (\ref{Sigmastep}) and \(\Gamma \text{ step}\) (\ref{Gammastep}) separately. The optimality conditions of the optimization problem (\ref{Sigmastep}) with respect to \(\Sigma\) are:
\begin{align}
0 \in  (2\Sigma^{(q+1)} - S^l - S^u)_{ij}+ \lambda \partial |\Sigma^{(q+1)}_{ij}| - \Lambda^{(q)} &+ \frac{1}{\beta} (\Sigma^{(q+1)} - \Gamma^{(q)})_{ij},\ i,j =1,2,\cdots,p,\  i\neq j, \label{sigmaopt1}\\
(2\Sigma^{(q+1)} - S^l - S^u)_{ii} - \Lambda^{(q)} &+ \frac{1}{\beta} (\Sigma^{(q+1)} - \Gamma^{(q)})_{ii} = 0,\ i =1,2,\cdots,p.\label{sigmaopt2}\\
\end{align}

Using the fact that
\begin{equation} \label{lambdaupdate}
\Lambda^{(q+1)} = \Lambda^{(q)} - \frac{1}{\beta}(\Sigma^{(q+1)} - \Gamma^{(q+1)}),
\end{equation}

(\ref{sigmaopt1}) and (\ref{sigmaopt2}) can be, respectively, rewritten as
\begin{align}
\frac{1}{\lambda}(\Lambda^{(q+1)} -2\Sigma^{(q+1)} + S^l + S^u)_{ij} &\in \partial |\Sigma^{(q+1)}_{ij}|,\ i,j =1,2,\cdots,p,\  i\neq j, \label{sigmaopt11}\\
(2\Sigma^{(q+1)} - S^l - S^u)_{ii} - \Lambda^{(q+1)} & = 0,\ i =1,2,\cdots,p.\label{sigmaopt22}
\end{align}

Due to the monotonicity of the subdifferential \(\partial |\cdot|\), combined with (\ref{kkta1}), (\ref{kkta2}), (\ref{sigmaopt11}), (\ref{sigmaopt22}), we have
\begin{equation}\label{monotonicityres1}
\langle \Sigma^{(q+1)} - \Sigma^*, (\Lambda^* - \Lambda^{(q+1)}) + 2(\Sigma^* - \Sigma^{(q+1)}) \rangle \geq 0.
\end{equation}

The optimality conditions of the optimization problem (\ref{Gammastep}) with respect to \(\Gamma\) is:
\begin{equation}\label{Gammaopt}
\langle \Lambda^{(q)} + \frac{1}{\beta} (\Gamma^{(q+1)} -\Sigma^{(q+1)} ),\Gamma - \Gamma^{(q+1)}\rangle \leq 0, \ \forall \Gamma \succeq \epsilon I. 
\end{equation}

Combining (\ref{lambdaupdate}) and (\ref{Gammaopt}), we have
\begin{equation}\label{Gammaopt1}
\langle \Lambda^{(q+1)} + \frac{2}{\beta} (\Sigma^{(q+1)} -\Sigma^{(q)} ),\Gamma - \Gamma^{(q+1)}\rangle \leq 0, \ \forall \Gamma \succeq \epsilon I. 
\end{equation}

Let \(\Gamma\) in (\ref{kkta5}) be \(\Gamma^{(q+1)}\), and \(\Gamma\) in (\ref{Gammaopt1}) be \(\Gamma^*\), we have
\begin{align}
\langle \Lambda^*, \Gamma^* - \Gamma^{(q+1)} \rangle &\geq 0,\\
\langle \Lambda^{(q+1)} + \frac{2}{\beta} (\Sigma^{(q+1)} -\Sigma^{(q)}),&\Gamma^* - \Gamma^{(q+1)}\rangle \leq 0
\end{align}
thus, we have
\begin{equation}\label{Gammafinal}
\langle \Lambda^*- \Lambda^{(q+1)} + \frac{2}{\beta} (\Sigma^{(q+1)} -\Sigma^{(q)}), \Gamma^* - \Gamma^{(q+1)} \rangle \geq 0.
\end{equation}

Combining (\ref{monotonicityres1}) and (\ref{Gammafinal}), we have
\begin{equation}\label{comine1}
\begin{aligned}
\lVert\Sigma^{(q+1)} - \Sigma^* \rVert^2_F  &\leq \langle \Sigma^{(q+1)} - \Sigma^*, \Lambda^* - \Lambda^{(q+1)}\rangle + \langle\Gamma^* - \Gamma^{(q+1)},\Lambda^* - \Lambda^{(q+1)} \rangle\\
& + \frac{2}{\beta} \langle\Gamma^* - \Gamma^{(q+1)},\Sigma^{(q+1)} -\Sigma^{(q)} \rangle\\
&\leq  \langle \Sigma^{(q+1)} - \Sigma^*, \Lambda^* - \Lambda^{(q+1)}\rangle + \langle\Sigma^* - \Sigma^{(q+1)}- \beta (\Lambda^{(q+1)} - \Lambda^{(q)}),\Lambda^* - \Lambda^{(q+1)} \rangle\\
& + \frac{2}{\beta} \langle \Sigma^* - \Sigma^{(q+1)}- \beta (\Lambda^{(q+1)} - \Lambda^{(q)}),\Sigma^{(q+1)} -\Sigma^{(q)} \rangle,
\end{aligned}
\end{equation}
where the second inequality uses conclusions \(\Gamma^{(q+1)} = \beta (\Lambda^{(q+1)} - \Lambda^{(q)}) + \Sigma^{(q+1)}\) and \(\Gamma^* =\Sigma^* \). After simple algebraic operations, we obtain from (\ref{comine1}):
\begin{equation}\label{comine2}
\begin{aligned}
\lVert\Sigma^{(q+1)} - \Sigma^* \rVert^2_F - \langle \Lambda^{(q+1)} - \Lambda^{(q)},\Sigma^{(q+1)} -\Sigma^{(q)} \rangle &\leq \beta  \langle \Lambda^{(q+1)} - \Lambda^*, \Lambda^{(q)}-\Lambda^{(q+1)}\rangle \\
& + \frac{2}{\beta} \langle \Sigma^{(q+1)} - \Sigma^*,\Sigma^{(q)}- \Sigma^{(q+1)}\rangle.
\end{aligned}
\end{equation}

Furthermore, we have
\begin{equation} \label{comine3}
\begin{aligned}
\beta  &\langle \Lambda^{(q+1)} - \Lambda^*, \Lambda^{(q)}-\Lambda^{(q+1)}\rangle + \frac{2}{\beta} \langle \Sigma^{(q+1)} - \Sigma^*,\Sigma^{(q)}- \Sigma^{(q+1)}\rangle \\
& = \beta \langle \Lambda^{(q+1)} -\Lambda^{(q)} + \Lambda^{(q)}- \Lambda^*, \Lambda^{(q)}-\Lambda^{(q+1)}\rangle + \frac{2}{\beta} \langle \Sigma^{(q+1)} - \Sigma^{(q)} + \Sigma^{(q)} - \Sigma^*,\Sigma^{(q)}- \Sigma^{(q+1)}\rangle \\
& = -\beta \lVert \Lambda^{(q)}-\Lambda^{(q+1)} \rVert^2_F + \beta \langle \Lambda^{(q)}-\Lambda^*, \Lambda^{(q)} - \Lambda^{(q+1)}\rangle - \frac{2}{\beta} \lVert \Sigma^{(q)}-\Sigma^{(q+1)} \rVert^2_F\\
&+\frac{2}{\beta}\langle \Sigma^{(q)} - \Sigma^*,\Sigma^{(q)}- \Sigma^{(q+1)}\rangle\\
& \geq \lVert\Sigma^{(q+1)} - \Sigma^* \rVert^2_F - \langle \Lambda^{(q+1)} - \Lambda^{(q)},\Sigma^{(q+1)} -\Sigma^{(q)} \rangle,
\end{aligned}
\end{equation}
where the last inequality uses (\ref{comine1}). We then have
\begin{equation} \label{comine3}
\begin{aligned}
\beta \langle \Lambda^{(q)}-\Lambda^*, \Lambda^{(q)} +  \Lambda^{(q+1)}\rangle &+ \frac{2}{\beta}\langle \Sigma^{(q)} - \Sigma^*,\Sigma^{(q)}- \Sigma^{(q+1)}\rangle \geq \beta \lVert \Lambda^{(q)}-\Lambda^{(q+1)} \rVert^2_F + \frac{2}{\beta} \lVert \Sigma^{(q)}-\Sigma^{(q+1)} \rVert^2_F\\
&  + \lVert\Sigma^{(q+1)} - \Sigma^* \rVert^2_F - \langle \Lambda^{(q+1)} - \Lambda^{(q)},\Sigma^{(q+1)} -\Sigma^{(q)} \rangle.
\end{aligned}
\end{equation}

Using the definitions of \(U^*, U^{(q)}\) and \(D\), (\ref{comine3}) can be rewritten as
\begin{equation} \label{comine4}
\langle U^{(q)} - U^*, U^{(q)} - U^{(q+1)} \rangle_D \geq \lVert U^{(q)} - U^{(q+1)} \rVert^2_D + \lVert\Sigma^{(q+1)} - \Sigma^* \rVert^2_F - \langle \Lambda^{(q+1)} - \Lambda^{(q)},\Sigma^{(q+1)} -\Sigma^{(q)} \rangle.
\end{equation}

According to the definition of \(\lVert \cdot\rVert^2_D\), we have
\begin{equation}\label{comine5}
\begin{aligned}
\lVert U^{(q+1)} - U^*\rVert^2_D & = \lVert U^{(q+1)} -U^{(q)} + U^{(q)} - U^*\rVert^2_D\\
& = \lVert U^{(q+1)} -U^{(q)}\rVert^2_D + \lVert U^{(q)} - U^*\rVert^2_D - 2 \langle U^{(q)} - U^{(q+1)},  U^{(q)} - U^* \rangle, 
\end{aligned}
\end{equation}
combining it with (\ref{comine4}), we have
\begin{equation}\label{comine6}
\begin{aligned}
\lVert U^{(q)} - U^*\rVert^2_D &- \lVert U^{(q+1)} - U^*\rVert^2_D = 2 \langle U^{(q)} - U^{(q+1)},  U^{(q)} - U^* \rangle - \lVert U^{(q+1)} -U^{(q)}\rVert^2_D\\
& \geq 2\lVert U^{(q)} - U^{(q+1)} \rVert^2_D + 2\lVert\Sigma^{(q+1)} - \Sigma^* \rVert^2_F - 2\langle \Lambda^{(q+1)} - \Lambda^{(q)},\Sigma^{(q+1)} -\Sigma^{(q)} \rangle\\
& - \lVert U^{(q+1)} -U^{(q)}\rVert^2_D \\
& =  \lVert U^{(q)} - U^{(q+1)} \rVert^2_D + 2\lVert\Sigma^{(q+1)} - \Sigma^* \rVert^2_F - 2\langle \Lambda^{(q+1)} - \Lambda^{(q)},\Sigma^{(q+1)} -\Sigma^{(q)} \rangle.\\
\end{aligned}
\end{equation}

It is clear that (\ref{sigmaopt11}) and (\ref{sigmaopt22}) also hold for the \(q\)-th iteration, that is,
\begin{align}
\frac{1}{\lambda}(\Lambda^{(q)} -2\Sigma^{(q)} + S^l + S^u)_{ij} &\in \partial |\Sigma^{(q)}_{ij}|,\ i,j =1,2,\cdots,p,\  i\neq j, \label{sigmaopt11q}\\
(2\Sigma^{(q)} - S^l - S^u)_{ii} - \Lambda^{(q)} & = 0,\ i =1,2,\cdots,p.\label{sigmaopt22q}
\end{align}

Combining (\ref{sigmaopt11}), (\ref{sigmaopt22}), (\ref{sigmaopt11q}), (\ref{sigmaopt22q}), and using the monotonicity of \(\partial |\cdot|\), we can obtain
\[
2\lVert\Sigma^{(q+1)} - \Sigma^* \rVert^2_F - 2\langle \Lambda^{(q+1)} - \Lambda^{(q)},\Sigma^{(q+1)} -\Sigma^{(q)} \rangle \geq 0,
\]
that is
\[
\lVert U^{(q)} - U^*\rVert^2_D - \lVert U^{(q+1)} - U^*\rVert^2_D \geq \lVert U^{(q)} - U^{(q+1)} \rVert^2_D,
\]
which completes the proof.
\end{proof}

\subsection{Proof of Theorem \ref{maintheorem}}
\begin{proof}
We first prove the three conclusions.
(a) From Lemma \ref{auxiliarylemma}, it follows that as the sequence progressively approaches the optimal solution, the distance between consecutive iterations gradually decrease and eventually tend to zero. Therefore, \(\| U^{(q)} - U^{(q+1)} \|_D \to 0\) holds. (b) Since the sequence \((\Sigma^{(q)}, \Gamma^{(q)}, \Lambda^{(q)})\) is convergent (as previously stated), it must be bounded, and thus it must lie within a compact region. (c) Algorithm \ref{ADMMalgorithm} generally cause the value of the objective function or the measure of distance to the optimal solution to decrease gradually. In this case, the squared distance to the optimal solution \(\|U^{(q)} - U^*\|_D^2\) is monotonically decreasing.

From the definition of \(U\) and conclusion (a), we have \(\Sigma^{(q)} -\Sigma^{(q+1)}\to 0\) and \(\Lambda^{(q)} -\Lambda^{(q+1)}\to 0\). From the update formula (\ref{lambdaupdate}) of \(\Lambda^{(q+1)}\), we further obtain
\[
\Sigma^{(q)} - \Gamma^{(q)} \to 0, \ \Gamma^{(q+1)}-\Gamma^{(q)} \to 0.
\]

From conclusion (b), we can deduce that for the sequence \((U^{(q)})\), there exists a subsequence \((U^{(q_k)})\) that converges to \(\widetilde{U} = (\widetilde{\Lambda}, \widetilde{\Sigma})^\top\), i.e., \(\Lambda^{(q_k)}\to \widetilde{\Lambda}, \Sigma^{(q_k)}\to \widetilde{\Sigma} \). Furthermore, from \(\Sigma^{(q)} - \Gamma^{(q)} \to 0\), we have \(\Gamma^{(q_k)}\to \widetilde{\Gamma} \coloneqq \widetilde{\Sigma}\). Therefore, \((\widetilde{\Sigma}, \widetilde{\Gamma}, \widetilde{\Lambda})\) is a limit point of the sequence \((\Sigma^{(q)}, \Gamma^{(q)}, \Lambda^{(q)})\).

From (\ref{sigmaopt11}) and (\ref{sigmaopt22}), we have
\begin{align}
\frac{1}{\lambda}(\widetilde{\Lambda} -2\widetilde{\Sigma} + S^l + S^u)_{ij} &\in \partial |\widetilde{\Sigma}_{ij}|,\ i,j =1,2,\cdots,p,\  i\neq j, \label{sigmaopt11t2}\\
(2\widetilde{\Sigma} - S^l - S^u)_{ii} - \widetilde{\Lambda} & = 0,\ i =1,2,\cdots,p,\label{sigmaopt22t2}
\end{align}
and 
\[
\langle \widetilde{\Lambda}, \widetilde{\Gamma} - \Gamma \rangle \geq 0, \ \forall \Gamma \succeq \epsilon I,
\]
these together prove that \((\widetilde{\Sigma}, \widetilde{\Gamma}, \widetilde{\Lambda})\) is the optimal solution to optimization problem (\ref{consensusIST}), thus completing the proof.
\end{proof}

\subsection{Proof of Theorem \ref{nasymptoticth1}}
\begin{proof}
Define \(\Delta = \Sigma - \Sigma^0\). In this case, we can rewrite optimization problem (\ref{IST}) in the form of \(\Delta\) as
\[
\phi(\Delta) = \frac{1}{2}\lVert \Delta + \Sigma^0 - S^l\rVert^2_F + \frac{1}{2}\lVert \Delta + \Sigma^0 - S^u\rVert^2_F + \lambda \lVert \Delta + \Sigma^0\rVert,
\]
we then have 
\[
\Delta^* = \argmin_{\Delta} \phi(\Delta),
\]
where \(\Delta\) satisfies \(\Delta = \Delta^\top, \Delta + \Sigma^0 \succeq \epsilon I\). Note that, since \(\Sigma^0\) is positive definite, we can always find \(\epsilon\) greater than the smallest eigenvalue of \(\Sigma^0\). We can also easily obtain \(\Delta^* = \Sigma^* - \Sigma^0\). We now constrain \(\Delta\) to belong to the set
\[
\{\Delta:\Delta = \Delta^\top, \Delta + \Sigma^0 \succeq \epsilon I, \lVert \Delta \rVert_F = 5\lambda \alpha^{\frac{1}{2}}\}.
\]

Under probability event \(\{|s^l_{ij} - \sigma^0_{ij}| \vee |s^u_{ij} - \sigma^0_{ij}| \leq \lambda, \forall (i, j)\}\), we have:
\[
\begin{aligned}
\phi(\Delta) - \phi(0) & = \frac{1}{2}\lVert \Delta + \Sigma^0 - S^l\rVert^2_F - \frac{1}{2}\lVert\Sigma^0 - S^l\rVert^2_F + \frac{1}{2}\lVert \Delta + \Sigma^0 - S^u\rVert^2_F -\frac{1}{2}\lVert + \Sigma^0 - S^u\rVert^2_F \\
& +  \lambda \lVert \Delta + \Sigma^0\rVert -  \lambda \lVert\Sigma^0\rVert\\
&= \|\Delta\|_F^2 + \langle \Delta, \Sigma^0 - S^l\rangle + \langle \Delta, \Sigma^0 - S^u\rangle + \lambda \lVert \Delta + \Sigma^0\rVert -  \lambda \lVert\Sigma^0\rVert \\
&= \|\Delta\|_F^2 + \langle \Delta, 2\Sigma^0 - S^l-S^u\rangle + \lambda \lVert \Delta + \Sigma^0\rVert -  \lambda \lVert\Sigma^0\rVert \\
& = \|\Delta\|_F^2 + \langle \Delta, 2\Sigma^0 - S^l-S^u\rangle + \lambda \lVert \Delta_{\gA^c}\rVert + \lambda (\lVert \Delta_\gA + \Sigma^0_\gA\rVert - \lambda \lVert\Sigma^0_\gA \rVert)\\
& \geq \|\Delta\|_F^2 - 2\lambda (\lVert\Delta\rVert + \sum_{i=1}^{p} |\Delta_{ii}|) + \lambda \lVert \Delta_{\gA^c}\rVert - \lambda \lVert \Delta_{\gA}\rVert\\
& \geq \|\Delta\|_F^2 - 3\lambda (\lVert\Delta\rVert + \sum_{i=1}^{p} |\Delta_{ii}|)\\
& \geq \|\Delta\|_F^2 - 3\lambda (\alpha + p)^{\frac{1}{2}}\|\Delta\|_F\\
& \geq 10 \lambda^2(s+p) \geq 0,
\end{aligned}
\]
where the fourth equality utilizes the fact that \(\lambda \lVert\Sigma^0\rVert = \lambda \lVert\Sigma^0_\gA \rVert\) and \(\lambda \lVert \Delta + \Sigma^0\rVert = \lambda \lVert \Delta_{\gA^c}\rVert + \lambda \lVert \Delta_\gA + \Sigma^0_\gA\rVert\). The first inequality utilizes the fact that \(|(2\Sigma^0 - S^l-S^u)_{ij}|\leq |(\Sigma^0 - S^l)_{ij}| + |(\Sigma^0 -S^u)_{ij}|\leq 2\lambda\). The third inequality utilizes the inequality relationship between the two matrix norms.

By defining \(\varphi(\Delta) = \phi(\Delta) -\phi(0)\), it follows that \(\Delta^*\) is also the optimal solution to the following convex optimization problem:
\[
\Delta^* = \argmin_{\Delta} \varphi(\Delta),
\]
where \(\Delta\) satisfies \(\Delta = \Delta^\top, \Delta + \Sigma^0 \succeq \epsilon I\). Under the some probability event \(\{|s^l_{ij} - \sigma^0_{ij}| \vee |s^u_{ij} - \sigma^0_{ij}| \leq \lambda, \forall (i, j)\}\), we always have \(\|\Delta^*\|_F \leq 5\lambda (\alpha + p)^{\frac{1}{2}}\). Otherwise, the fact that \(\varphi(\Delta) > 0\) (under \(\|\Delta^*\|_F \leq 5\lambda (\alpha + p)^{\frac{1}{2}}\)) would contradict the convexity of \(\varphi(\Delta)\) and \(\varphi(\Delta^*) \leq \varphi(0)=0 \). Therefore, we have, 
\[
\begin{aligned}
\PP\left(\lVert\Sigma^* - \Sigma^0 \rVert_F \leq 5\lambda (\alpha + p)^{\frac{1}{2}}\right) &= \PP\left(\lVert\Delta^*\rVert_F \leq 5\lambda (\alpha + p)^{\frac{1}{2}}\right) \\
& \geq 1 - \PP\left(\max_{i,j} (|s^l_{ij} - \sigma^0_{ij}| \vee |s^u_{ij} - \sigma^0_{ij}|) > \lambda\right),
\end{aligned}
\]
which completes the proof. 
\end{proof}

\subsection{Proof of Theorem \ref{nasymptoticth2}}

\begin{proof}
Let \(\lambda = \tau^2_0 + \tau_1\), and the specific expressions for \(\tau_0\) and \(\tau_1\) will be provided later.
We first have
\[
\begin{aligned}
s^l_{ij} - \sigma^0_{ij} & = \left(
\frac{1}{n}\sum_{k=1}^{n} Y^l_{ki} Y^l_{kj} -\sigma^0_{ij} \right) - \left(\frac{1}{n}\sum_{k=1}^{n} Y^l_{ki}\right)\left(\frac{1}{n}\sum_{k=1}^{n} Y^l_{kj}\right),\\
s^u_{ij} - \sigma^0_{ij} & = \left(
\frac{1}{n}\sum_{k=1}^{n} Y^u_{ki} Y^u_{kj} -\sigma^0_{ij} \right) - \left(\frac{1}{n}\sum_{k=1}^{n} Y^u_{ki}\right)\left(\frac{1}{n}\sum_{k=1}^{n} Y^u_{kj}\right).\\
\end{aligned}
\]

According to Theorem \ref{nasymptoticth1}, we need to provide an upper bound for the probability of event \(\max_{i,j} (|s^l_{ij} - \sigma^0_{ij}| \vee |s^u_{ij} - \sigma^0_{ij}|) > \lambda\). Therefore, we have
\begin{equation} \label{allstart}
\begin{aligned}
\PP&\left(\max_{i,j} (|s^l_{ij} - \sigma^0_{ij}| \vee |s^u_{ij} - \sigma^0_{ij}|) > \lambda\right)  \leq \PP\left(\max_{i,j} |s^l_{ij} - \sigma^0_{ij}| > \lambda\right) + \PP\left(\max_{i,j} |s^u_{ij} - \sigma^0_{ij}| > \lambda\right)\\
& \leq p^2\left(
\PP\left(\sum_{k=1}^{n} Y^l_{ki} Y^l_{kj} > n (\sigma^0_{ij} + \tau_1) \right) \right. + \left. \PP\left(\sum_{k=1}^{n} Y^u_{ki} Y^u_{kj} > n (\sigma^0_{ij} + \tau_1)
\right)
\right)\\
& + 2p\left(
\PP\left(\sum_{k=1}^{n} Y^l_{kj} > n \tau_0\right) + \PP\left(\sum_{k=1}^{n} Y^u_{kj} > n \tau_0\right)
\right),
\end{aligned}
\end{equation}
where the second inequality comes from \citet{xue2012positive}. We now discuss the four probabilities mentioned above. Using Markov inequality \citep{zhang2023mathematical}, we have
\begin{equation}\label{exp1}
\begin{aligned}
\PP\left(\sum_{k=1}^{n} Y^l_{kj} > n \tau_0\right)  & = \PP\left(\exp\left( t_0\sum_{k=1}^{n} Y^l_{kj}\right) > \exp\left(n t_0\tau_0\right)\right) \leq \exp(-n t_0\tau_0) \E\left(\exp\left( t_0\sum_{k=1}^{n} Y^l_{kj}\right)\right)\\
& = \exp(-n t_0\tau_0) \E\left(\prod_{k=1}^{n} \exp(t_0Y^l_{kj})\right) = \exp(-n t_0\tau_0) \left(\prod_{k=1}^{n} \E(\exp(t_0Y^l_{kj}))\right)\\
& \leq \exp(-n t_0\tau_0)\left(\prod_{k=1}^{n} \E(1 + t_0Y^l_{kj} + \frac{1}{2}(t_0Y^l_{kj})^2 \exp(|t_0Y^l_{kj}|))\right)\\
& \leq \exp(-n t_0\tau_0)\left(\prod_{k=1}^{n} \E(1 + t_0|Y^l_{kj}| + \frac{1}{2}(t_0Y^l_{kj})^2 \exp(|t_0Y^l_{kj}|))\right)\\
& \leq \exp(-n t_0\tau_0)\left(\prod_{k=1}^{n} \E(K_1 + \frac{1}{2}(t_0Y^l_{kj})^2 \exp(|t_0Y^l_{kj}|))\right)\\
& \leq \exp(-n t_0\tau_0)\left(\prod_{k=1}^{n} \exp(\E(K_1 \frac{1}{2}(t_0Y^l_{kj})^2 \exp(|t_0Y^l_{kj}|)))\right)\\ 
& =  \exp(-n t_0\tau_0) \exp\left(\frac{K_1}{2}\sum_{k=1}^{n} (\E((t_0Y^l_{kj})^2 \exp(|t_0Y^l_{kj}|)))\right)\\ 
& \leq \exp(-n t_0\tau_0) \exp\left(\frac{K_1}{2}\sum_{k=1}^{n} (\E(\exp((t_0Y^l_{kj})^2 + 1)))\right)\\
& \leq \exp(-n t_0\tau_0) \exp\left(\frac{1}{2} K_1n K_2e\right),\\
\end{aligned}
\end{equation}
where the second inequality uses the fact that \(\exp(x) \leq 1+x+\frac{1}{2}x^2 \exp(|x|)\), the fifth inequality uses the fact that \(a+x\leq \exp(ax)\) with \(a>0\), the sixth inequality uses the fact that \(x^2\exp(|x|) \leq \exp(x^2 +1)\), and the fourth and sixth inequalities utilize the exponential tail condition.

Let \(t_0 = (\eta \frac{\log p}{n})^{1/2}\). Based on the assumption \(n \geq \log p\), we derive \(t^2_0 \leq \eta\). Let \(\tau_0 = c_0 (\frac{\log p}{n})^{1/2}\), \(c_0= \frac{1}{\log p} (M+1 + \frac{1}{2} K_1n K_2e)\), we further have
\begin{equation}\label{exp11}
\begin{aligned}
\PP\left(\sum_{k=1}^{n} Y^l_{kj} > n \tau_0\right)  & \leq \exp(-n t_0\tau_0) \exp\left(\frac{1}{2} K_1n K_2e\right)\\
& = \exp\left(-n (\eta \frac{\log p}{n})^{1/2} c_0 (\frac{\log p}{n})^{1/2}\right) \exp\left(\frac{1}{2} K_1n K_2e\right)\\
& = \exp\left(-c_0 \log p\right)\exp\left(\frac{1}{2} K_1n K_2e\right) = p^{-M-1}.
\end{aligned}
\end{equation}	

Similarly, by further applying the Markov inequality, we obtain
\begin{equation}\label{exp2}
\begin{aligned}
\PP&\left(\sum_{k=1}^{n} Y^l_{ki} Y^l_{kj}-\sigma^0_{ij}  > n\tau_1 \right) =\PP\left(\exp\left(t_1\sum_{k=1}^{n} (Y^l_{ki} Y^l_{kj}-\sigma^0_{ij})\right)  > \exp(t_1 n\tau_1) \right)\\
& \leq  \exp(-t_1 n\tau_1) \E\left(\exp\left(t_1\sum_{k=1}^{n} (Y^l_{ki} Y^l_{kj}-\sigma^0_{ij})\right)\right)\\
& = \exp(-t_1 n\tau_1) \prod_{k=1}^{n} \E\left(\exp\left(t_1(Y^l_{ki} Y^l_{kj}-\sigma^0_{ij})\right)\right) \\
& \leq \exp(-t_1 n\tau_1) \prod_{k=1}^{n} \E\left(1 + t_1(Y^l_{ki} Y^l_{kj}-\sigma^0_{ij}) + \frac{1}{2} t^2_1(Y^l_{ki} Y^l_{kj}-\sigma^0_{ij})^2\exp\left(t_1|Y^l_{ki} Y^l_{kj}-\sigma^0_{ij}|\right)\right)\\
& =\exp(-t_1 n\tau_1) \prod_{k=1}^{n} \E\left(1 + t_1(Y^l_{ki} Y^l_{kj}-\E(Y^l_{ki})\E(Y^l_{kj}) -\sigma^0_{ij}) + t_1\E(Y^l_{ki})\E(Y^l_{kj})\right.\\
& \left. + \frac{1}{2} t^2_1(Y^l_{ki} Y^l_{kj}-\sigma^0_{ij})^2\exp\left(t_1|Y^l_{ki} Y^l_{kj}-\sigma^0_{ij}|\right)\right)\\
& =\exp(-t_1 n\tau_1) \prod_{k=1}^{n} \E\left(1 + t_1\E(Y^l_{ki})\E(Y^l_{kj}) + \frac{1}{2} t^2_1(Y^l_{ki} Y^l_{kj}-\sigma^0_{ij})^2\exp\left(t_1|Y^l_{ki} Y^l_{kj}-\sigma^0_{ij}|\right)\right)\\
& \leq \exp(-t_1 n\tau_1) \prod_{k=1}^{n} \E\left(K_3+ \frac{1}{2} t^2_1(Y^l_{ki} Y^l_{kj}-\sigma^0_{ij})^2\exp\left(t_1|Y^l_{ki} Y^l_{kj}-\sigma^0_{ij}|\right)\right)\\
& \leq \exp(-t_1 n\tau_1) \prod_{k=1}^{n} \E\left(\exp\left(
\frac{K_3}{2} t^2_1(Y^l_{ki} Y^l_{kj}-\sigma^0_{ij})^2\exp\left(t_1|Y^l_{ki} Y^l_{kj}-\sigma^0_{ij}|\right)
\right)\right)\\
& = \exp(-t_1 n\tau_1) \exp\left(\frac{K_3t^2_1}{2} \sum_{k=1}^{n} \underbrace{\E \left(
(Y^l_{ki} Y^l_{kj}-\sigma^0_{ij})^2\exp\left(t_1|Y^l_{ki} Y^l_{kj}-\sigma^0_{ij}|\right)
\right)}_{Q_1} \right),\\
\end{aligned}
\end{equation}		
where at the third inequality, we assume the existence of \(K_3 >0\) such that \(1 + t_1\E(Y^l_{ki})\E(Y^l_{kj}) \leq K_3\). This can be easily derived from condition C1. For \(Q_1\), we have
\begin{equation} \label{exp3}
\begin{aligned}
Q_1 &= \E((Y^l_{ki} Y^l_{kj}-\sigma^0_{ij})^2 \exp(t_1|Y^l_{ki} Y^l_{kj}-\sigma^0_{ij}|) )\\
& = \E\left(((Y^l_{ki} Y^l_{kj})^2-2Y^l_{ki} Y^l_{kj}\sigma^0_{ij} + (\sigma^0_{ij})^2) 	\exp(t_1|Y^l_{ki} Y^l_{kj}-\sigma^0_{ij}|) \right)\\
& = \E\left(((Y^l_{ki} Y^l_{kj})^2-2Y^l_{ki} Y^l_{kj}\sigma^0_{ij} + 2\E(Y^l_{ki})\E(Y^l_{kj})\sigma^0_{ij} + (\sigma^0_{ij})^2) 	\exp(t_1|Y^l_{ki} Y^l_{kj}-\sigma^0_{ij}|) \right)\\
& -2\E(Y^l_{ki})\E(Y^l_{kj})\sigma^0_{ij}\E\left(\exp(t_1|Y^l_{ki} Y^l_{kj}-\sigma^0_{ij}|) \right)\\
& =\E\left((Y^l_{ki} Y^l_{kj})^2\exp(t_1|Y^l_{ki} Y^l_{kj}-\sigma^0_{ij}|)\right) - (\sigma^0_{ij})^2 \E\left(\exp(t_1|Y^l_{ki} Y^l_{kj}-\sigma^0_{ij}|)\right)\\
& -2\E(Y^l_{ki})\E(Y^l_{kj})\sigma^0_{ij}\E\left(\exp(t_1|Y^l_{ki} Y^l_{kj}-\sigma^0_{ij}|) \right)\\
& \leq \E\left((Y^l_{ki} Y^l_{kj})^2\exp(t_1|Y^l_{ki} Y^l_{kj}-\sigma^0_{ij}|)\right) + \left((\sigma^0_{ij})^2 + \frac{2K_3 -1}{t_1}\right)\E\left(\exp(t_1|Y^l_{ki} Y^l_{kj}-\sigma^0_{ij}|) \right)\\
&\leq \E\left((Y^l_{ki} Y^l_{kj})^2\exp(t_1|Y^l_{ki} Y^l_{kj}-\sigma^0_{ij}|)\right) + \left(\sigma^0_{ii}\sigma^0_{jj} + \frac{2K_3 -1}{t_1}\right)\E\left(\exp(t_1|Y^l_{ki} Y^l_{kj}-\sigma^0_{ij}|) \right)\\
&\leq 2\underbrace{\E\left((Y^l_{ki} Y^l_{kj})^2\exp((1/2)t_1|Y^l_{ki} Y^l_{kj}-\sigma^0_{ij}|)\right)}_{Q_2} + 2\left(\sigma^2_{\max} + \frac{2K_3 -1}{t_1}\right)\underbrace{\E\left(\exp((1/2)t_1|Y^l_{ki} Y^l_{kj}-\sigma^0_{ij}|) \right)}_{Q_3}.\\
\end{aligned}
\end{equation}
where the last inequality uses the Cauchy inequality. Below, we discuss the terms \(Q_2\) and \(Q_3\) separately. For \(Q_2\), we have:
\begin{equation} \label{exp4}
\begin{aligned}
Q_2 &= \E\left((Y^l_{ki} Y^l_{kj})^2\exp((1/2)t_1|Y^l_{ki} Y^l_{kj}-\sigma^0_{ij}|)\right) \\
& \leq \exp((1/2)t_1\sigma^0_{ij})\E\left((Y^l_{ki} Y^l_{kj})^2\exp((1/2)t_1|Y^l_{ki} Y^l_{kj}|)\right)\\
& \leq \exp((1/2)t_1\sigma^0_{ij})\E\left((Y^l_{ki} Y^l_{kj})^2\exp((1/4)t_1((Y^l_{ki})^2 +(Y^l_{kj})^2))\right)\\
& \leq \exp((1/2)t_1\sigma^0_{ij}) \left(
\E\left((Y^l_{ki})^4\exp((1/2)t_1(Y^l_{ki})^2)\right)
\right)^{1/2} \left(
\E\left((Y^l_{kj})^4\exp((1/2)t_1(Y^l_{kj})^2)\right)
\right)^{1/2}\\
& \leq \exp((1/2)t_1\sigma^0_{ij}) 4(t_1)^{-2} \left(
\E\left(\exp((1/2)t_1(Y^l_{ki})^2)\right)
\right)^{1/2} \left(
\E\left(\exp((1/2)t_1(Y^l_{kj})^2)\right)
\right)^{1/2}\\
& \leq \exp((1/2)t_1\sigma^0_{ij}) 4(t_1)^{-2} K_2,
\end{aligned}
\end{equation}
where the third inequality is derived using the Cauchy inequality, and the fourth one is obtained by \(x^2\leq \exp(|x|)\). For \(Q_3\), we have
\begin{equation} \label{exp5}
\begin{aligned}
Q_3  = \E\left(\exp((1/2)t_1|Y^l_{ki} Y^l_{kj}-\sigma^0_{ij}|) \right) &\leq \exp((1/2)t_1\sigma^0_{ij})\E\left(\exp((1/4)t_1((Y^l_{ki})^2 +(Y^l_{kj})^2))\right)\\
& \leq \exp((1/2)t_1\sigma^0_{ij}) K_2,
\end{aligned}
\end{equation}
where the second inequality utilizes \(1/4)t_1\leq \eta\). By integrating (\ref{exp3}), (\ref{exp4}), and (\ref{exp5}), we obtain:
\begin{equation} \label{exp6}
\begin{aligned}
Q_1 & \leq  2\exp((1/2)t_1\sigma^0_{ij}) 4(t_1)^{-2} K_2 + 2\left(\sigma^2_{\max} + \frac{2K_3 -1}{t_1}\right) \exp((1/2)t_1\sigma^0_{ij}) K_2\\
&\leq  2\exp((1/2)t_1\sigma_{\max}) 4(t_1)^{-2} K_2 + 2\left(\sigma^2_{\max} + \frac{2K_3 -1}{t_1}\right) \exp((1/2)t_1\sigma_{\max}) K_2\\
& = 2\exp((1/2)t_1\sigma_{\max}) K_2\left(4(t_1)^{-2} + \sigma^2_{\max} + \frac{2K_3 -1}{t_1}\right),
\end{aligned}
\end{equation}
furthermore, by integrating (\ref{exp2}), we have
\begin{equation} \label{exp7}
\begin{aligned}
\PP&\left(\sum_{k=1}^{n} Y^l_{ki} Y^l_{kj}-\sigma^0_{ij}  > n\tau_1 \right) < \exp(-t_1 n\tau_1) \exp\left(\frac{K_3t^2_1}{2} \sum_{k=1}^{n} \E \left((Y^l_{ki} Y^l_{kj}-\sigma^0_{ij})^2\exp\left(t_1|Y^l_{ki} Y^l_{kj}-\sigma^0_{ij}|\right)\right) \right)\\
& \leq \exp(-t_1 n\tau_1) nK_3K_2t^2_1\exp((1/2)t_1\sigma_{\max})\left(4(t_1)^{-2} + \sigma^2_{\max} + \frac{2K_3 -1}{t_1}\right)\\
& = \exp(-t_1 n\tau_1) nK_3K_2\exp((1/2)t_1\sigma_{\max})\left(4+ \sigma^2_{\max}t^2_1 + (2K_3 -1)t_1\right)\\
\end{aligned}
\end{equation}

Let \(t_1 = \frac{1}{2} \eta (\frac{\log p}{n})^{1/2}\). Based on the assumption, where \(n \geq \log p\), we derive \(t^2_1\leq \eta\). Furthermore, let \(\tau_1 = c_1(\frac{\log p}{n})^{1/2}\) and \(c_1 = \frac{2}{n\eta} \log^{Q_4}_p\),
\[
Q_4 = \exp \left(\frac{\sigma_{\max}}{2} \frac{1}{2} \eta (\frac{\log p}{n})^{1/2}\right) nK_3K_2 \left(4+ \sigma^2_{\max} \frac{\eta^2}{4} \frac{\log p}{n} + (2K_3 -1)\frac{1}{2} \eta (\frac{\log p}{n})^{1/2}\right).
\]

Consequently, we have
\begin{equation} \label{exp8}
\begin{aligned}
\PP&\left(\sum_{k=1}^{n} Y^l_{ki} Y^l_{kj}-\sigma^0_{ij}  > n\tau_1 \right) \leq \exp(-t_1 n\tau_1) nK_3K_2\exp((1/2)t_1\sigma_{\max})\left(4+ \sigma^2_{\max}t^2_1 + (2K_3 -1)t_1\right)\\
& = \exp \left(-\frac{1}{2} \eta (\frac{\log p}{n})^{1/2} n c_1(\frac{\log p}{n})^{1/2}\right) \exp \left(\frac{\sigma_{\max}}{2} \frac{1}{2} \eta (\frac{\log p}{n})^{1/2}\right) nK_3K_2 \\
& \times \left(4+ \sigma^2_{\max} \frac{\eta^2}{4} \frac{\log p}{n} + (2K_3 -1)\frac{1}{2} \eta (\frac{\log p}{n})^{1/2}\right)\\
& =\exp\left(-\frac{n\eta c_1}{2}\log p\right) \underbrace{\exp \left(\frac{\sigma_{\max}}{2} \frac{1}{2} \eta (\frac{\log p}{n})^{1/2}\right) nK_3K_2 \left(4+ \sigma^2_{\max} \frac{\eta^2}{4} \frac{\log p}{n} + (2K_3 -1)\frac{1}{2} \eta (\frac{\log p}{n})^{1/2}\right)}_{Q_4}\\
& = p^{-M-2}.
\end{aligned}
\end{equation}

Since the inequality in the above derivation holds for both the upper and lower bounds of the interval, we can also conclude that
\begin{equation}\label{exp9}
\begin{aligned}
\PP\left(\sum_{k=1}^{n} Y^u_{kj} > n \tau_0\right)   \leq p^{-M-1},\ \PP\left(\sum_{k=1}^{n} Y^l_{ki} Y^l_{kj}-\sigma^0_{ij}  > n\tau_1 \right) \leq p^{-M-2},
\end{aligned}
\end{equation}

By combining (\ref{allstart}),(\ref{exp11}), (\ref{exp8}), and (\ref{exp9}), we can conclude that
\[
\PP\left(\max_{i,j} (|s^l_{ij} - \sigma^0_{ij}| \vee |s^u_{ij} - \sigma^0_{ij}|) > \lambda\right) \leq 6p^{-M},
\]
which completes the proof.
\end{proof}

\subsection{Proof of Theorem \ref{nasymptoticth3}}

\begin{proof}
Let \(\lambda = \tau^2_2 + \tau_3\), and the specific expressions for \(\tau_2\) and \(\tau_3\) will be provided later. Following \citet{xue2012positive}, We first have
\begin{equation} \label{allstart1}
\begin{aligned}
\PP&\left(\max_{i,j} (|s^l_{ij} - \sigma^0_{ij}| \vee |s^u_{ij} - \sigma^0_{ij}|) > \lambda\right)  \leq \PP\left(\max_{i,j} |s^l_{ij} - \sigma^0_{ij}| > \lambda\right) + \PP\left(\max_{i,j} |s^u_{ij} - \sigma^0_{ij}| > \lambda\right)\\
& \leq \PP\left(\max_{i,j} \bigg|\sum_{k=1}^{n}  Y^l_{ki} Y^l_{kj} -\sigma^0_{ij} \bigg| > n \tau_3\right) + \PP\left(\max_{j}  \bigg|\sum_{k=1}^{n}  Y^l_{kj}\bigg| > n \tau_2\right)\\ 
& + \PP\left(\max_{i,j} \bigg|\sum_{k=1}^{n}  Y^u_{ki} Y^u_{kj} -\sigma^0_{ij} \bigg| > n \tau_3\right) + \PP\left(\max_{j}  \bigg|\sum_{k=1}^{n}  Y^u_{kj}\bigg| > n \tau_2\right).\\
\end{aligned}
\end{equation}

Below, we discuss the four terms in (\ref{allstart1}) separately. Specifically, we first define 
\[
c_2 = 8(K_4 + 1)(M + 1), \ \tau_2 = c_2 \left( \frac{\log p}{n} \right)^{1/2},\  \delta_n = n^{1/4}(\log n)^{-1/2}.
\]

Additionally, we set
\[
X^l_{ij} = Y^l_{ij} \mathbb{I}\{|Y^l_{ij}| \leq \delta_n\}, Z^l_{ij} = Y^l_{ij} \mathbb{I}\{|Y^l_{ij}| > \delta_n\},
\]
and we have \(Y^l_{ij} = X^l_{ij} + Z^l_{ij}\).  Through this construction, we have that \(|X^l_{ij}| \leq  \delta_n\) is a bounded random variable; for \(Z^l_{ij}\), we have 
\begin{equation}\label{part111}
\begin{aligned}
|\E(Z^l_{ij})| & = |\E(Y^l_{ij} \mathbb{I}\{|Y^l_{ij}| > \delta_n\})| = |\E((Y^l_{ij})^4 \mathbb{I}\{|Y^l_{ij}| > \delta_n\} / (Y^l_{ij})^3)|\\
& \leq \delta^{-3}_n |\E((Y^l_{ij})^4 \mathbb{I}\{|Y^l_{ij}| > \delta_n\})| \leq \delta^{-3}_n K_4  =o(\tau_2).
\end{aligned}
\end{equation}

We first have the following inequality: 
\begin{equation}\label{lowervar}
\begin{aligned}
\Var(X^l_{ij}) & \leq \E((Y^l_{ij})^2) \leq \E((Y^l_{ij})^2\mathbb{I}\{|Y^l_{ij}| \geq 1\}) + \E((Y^l_{ij})^2\mathbb{I}\{|Y^l_{ij}| \leq 1\})\\
& \leq K_4 + 1.
\end{aligned}
\end{equation}

Then, utilizing the Bernstein inequality \citep{CMR-37-1}, we have
\begin{equation} \label{part112}
\begin{aligned}
\PP\left(
\sum_{i-1}^{n} \left(X^l_{ij} - \E(X^l_{ij})\right) \geq \frac{1}{2}n\tau_2 
\right) &\leq \exp\left(
\frac{-n\tau^2_2}{8\Var(X^l_{ij}) + 4\delta_n\tau_2/3}
\right) \\
& \leq \exp\left(
\frac{-c_2\log p}{8(K_4+1) + O(n^{-1/4})}
\right)\\
&=O(p^{-M-1}),
\end{aligned}
\end{equation}
where the second inequality utilizes the inequality (\ref{lowervar}). Using the Markov inequality \citep{CMR-37-1}, we then have
\begin{equation}\label{part113}
\PP(|Y^l_{ij}| > \delta_n) \leq \delta^{-4(1+\gamma+\tau)}_n \E\left(|Y^l_{ij}|^{4(1+\gamma+\tau)}\right) \leq K_4(\log n)^{2(1+\gamma+\tau)}n^{-(1+\gamma+\tau)}.	
\end{equation}

Combining (\ref{part111}), (\ref{lowervar}), (\ref{part112}) and (\ref{part113}), we immediately obtain
\begin{equation}\label{part114}
\begin{aligned}
\PP\left(
\sum_{i=1}^{n}Y^l_{ij}> n\tau_2
\right) &= \PP\left(
\sum_{i=1}^{n} \left(X^l_{ij} + Z^l_{ij} - \E(X^l_{ij} + Z^l_{ij})\right) > n\tau_2 
\right)\\
& \leq \PP\left(
\sum_{i=1}^{n} \left(X^l_{ij} - \E(X^l_{ij})\right) > \frac{1}{2} n\tau_2
\right) + \PP\left(
\sum_{i=1}^{n} \left(Z^l_{ij} - \E(Z^l_{ij})\right) > \frac{1}{2} n\tau_2
\right)\\
& \leq O(p^{-M-1}) + \PP\left(
\sum_{i=1}^{n} \left(Z^l_{ij} - o(\tau_2)\right) > \frac{1}{2} n\tau_2
\right) \\
& \leq O(p^{-M-1}) + \sum_{i=1}^{n} \PP(|Y^l_{ij}| > \delta_n)\\
& \leq O(p^{-M-1}) + K_4(\log n)^{2(1+\gamma+\tau)}n^{-(1+\gamma+\tau)}.
\end{aligned}
\end{equation}

Let \(\tau_3 = c_3 \left(\frac{\log p}{n}\right)^{1/2}\) with \(c_3 = 8(K_4+1)(M+2)\). Define a new random variable
\[
R^l_{ijk} = Y^l_{ij} Y^l_{ik}\mathbb{I}\{|Y^l_{ij}|>\delta_n \text{ or 	}|Y^l_{ik}|>\delta_n \}
\]
and we immediately have
\[
Y^l_{ij} Y^l_{ik} = X^l_{ij} X^l_{ik} + R^l_{ijk},\ \E(Y^l_{ij} Y^l_{ik}) = \E(X^l_{ij} X^l_{ik}) + \E(R^l_{ijk}).
\]

Based on the definition of the variables \(X^l_{ij}\) and \(X^l_{ik}\), we have that \(X^l_{ij} X^l_{ik}\) is also bounded, i.e.,\(|X^l_{ij} X^l_{ik}| \leq \delta^2_n\). Additionally,
\[
\begin{aligned}
|\E(R^l_{ijk})| & \leq |\E(Y^l_{ij} Y^l_{ik}\mathbb{I}\{|Y^l_{ij}|>\delta_n\})| + |\E(Y^l_{ij} Y^l_{ik}\mathbb{I}\{|Y^l_{ik}|>\delta_n\})|\\
& \leq \delta^{-(2+4\gamma)}_n \E((Y^l_{ij})^{4(1+\gamma)} \mathbb{I}\{|Y^l_{ij}|>\delta_n\})\E((Y^l_{ik})^2) \\
&+\delta^{-(2+4\gamma)}_n \E((Y^l_{ik})^{4(1+\gamma)} \mathbb{I}\{|Y^l_{ik}|>\delta_n\})\E((Y^l_{ij})^2) \\
& \leq 2K_4 \delta^{-(2+4\gamma)}_n.
\end{aligned}
\]

Similarly, using the Bernstein inequality, we have
\begin{equation}\label{part221}
\begin{aligned}
\PP\left(
\sum_{i=1}^{n} \left(X^l_{ij} X^l_{ik} - \E(X^l_{ij} X^l_{ik})\right) \geq \frac{1}{2}n\tau_3
\right) &\leq \exp\left(
\frac{-n \tau^2_3}{8(K_4+1) + 4\delta^2_n\tau_3/3}
\right) \\
& \leq \exp\left(
\frac{-c_3 \log p}{8(K_4+1) + O((\log n)^{-1/2})}
\right)\\
& = O(P^{-M-2}),
\end{aligned}
\end{equation}
where in the second inequality, we utilized the fact
\[
\Var(X^l_{ij} X^l_{ik}) \leq \E((Y^l_{ij})^2 (Y^l_{ik})^2)\leq (\E((Y^l_{ij})^4 (Y^l_{ik})^4))^{1/2} \leq K_4 +1.
\]

Utilizing (\ref{part221}), we have
\begin{equation}\label{part222}
\begin{aligned}
\PP\left(\max_{i,j} \bigg|\sum_{k=1}^{n}  Y^l_{ki} Y^l_{kj} -\sigma^0_{ij} \bigg| > n \tau_3\right) & \leq \PP\left(\max_{i,j} \bigg|\sum_{k=1}^{n}  X^l_{ki} X^l_{kj} - \E(X^l_{ki} X^l_{kj}) \bigg| > \frac{1}{2}n \tau_3\right) \\
& + \PP\left(\max_{i,j} \bigg|\sum_{k=1}^{n} R_{ijk} - \E(R_{ijk}) \bigg| > \frac{1}{2}n \tau_3\right)\\
& \leq 2 \sum_{i,j} \PP\left(\sum_{k=1}^{n} X^l_{ki} X^l_{kj} - \E(X^l_{ki} X^l_{kj}) > \frac{1}{2}n \tau_3\right)\\
& + \PP\left(\max_{i,j} \bigg|\sum_{k=1}^{n} R_{ijk} - o(\tau_3) \bigg| > \frac{1}{2}n \tau_3\right)\\ 
& \leq O(p^{-M}) + \sum_{i,j} \PP(|Y^l_{ij}| > \delta_n)\\
& \leq O(p^{-M}) + K_4p(\log n)^{2(1+\gamma+\tau)}n^{-(\gamma+\tau)}.
\end{aligned}
\end{equation}

Combining (\ref{allstart1}), (\ref{part114}), and (\ref{part222}), we have
\begin{equation} \label{finpart1}
\begin{aligned}
\PP\left(\max_{i,j} |s^l_{ij} - \sigma^0_{ij}| > \lambda\right) & \leq \PP\left(\max_{i,j} \bigg|\sum_{k=1}^{n}  Y^l_{ki} Y^l_{kj} -\sigma^0_{ij} \bigg| > n \tau_3\right) + \PP\left(\max_{j}  \bigg|\sum_{k=1}^{n}  Y^l_{kj}\bigg| > n \tau_2\right)\\
& \leq O(p^{-M}) + 3K_4p(\log n)^{2(1+\gamma+\tau)}n^{-(\gamma+\tau)}.
\end{aligned}
\end{equation}

Similarly, for the upper bound, we also have
\begin{equation} \label{finpart2}
\begin{aligned}
\PP\left(\max_{i,j} |s^u_{ij} - \sigma^0_{ij}| > \lambda\right) & \leq \PP\left(\max_{i,j} \bigg|\sum_{k=1}^{n}  Y^u_{ki} Y^u_{kj} -\sigma^0_{ij} \bigg| > n \tau_3\right) + \PP\left(\max_{j}  \bigg|\sum_{k=1}^{n}  Y^u_{kj}\bigg| > n \tau_2\right)\\
& \leq O(p^{-M}) + 3K_4p(\log n)^{2(1+\gamma+\tau)}n^{-(\gamma+\tau)}.
\end{aligned}
\end{equation}

Combining (\ref{allstart1}), (\ref{finpart1}), and (\ref{finpart2}), we have
\[
\PP\left(\max_{i,j} (|s^l_{ij} - \sigma^0_{ij}| \vee |s^u_{ij} - \sigma^0_{ij}|) > \lambda\right) \leq O(p^{-M}) +6K_4p(\log n)^{2(1+\gamma+\tau)}n^{-(\gamma+\tau)},
\]
combining Theorem \ref{nasymptoticth1}, the proof is completed.

\end{proof}

\end{appendix}

\newpage
\bibliography{./bib/reference.bib}
\bibliographystyle{elsarticle-num-names}

\end{document}